\documentclass[10pt,journal]{IEEEtran}

\makeatletter
\def\ps@headings{%
\def\@oddhead{\mbox{}\scriptsize\rightmark \hfil \thepage}%
\def\@evenhead{\scriptsize\thepage \hfil \leftmark\mbox{}}%
\def\@oddfoot{}%
\def\@evenfoot{}}

\makeatother

\usepackage{graphicx}
\usepackage[noadjust]{cite}
\usepackage{tabularx,ragged2e,booktabs,caption}
\usepackage{times}
\usepackage{alltt}
\usepackage{verbatim}
\usepackage{moreverb}
\usepackage{amsmath}
\usepackage{amssymb}
\usepackage{url}
\usepackage{epsfig}
\usepackage{epstopdf}
\usepackage{subfigure}

\usepackage{multirow}
\usepackage[ruled,vlined,linesnumbered]{algorithm2e}
\usepackage{mathrsfs}
\usepackage{amsthm}
\usepackage{listings}
\usepackage[none]{hyphenat}
\usepackage{xcolor}
\usepackage{caption}
\usepackage{float}

\usepackage{tabularx}
\usepackage{adjustbox}

\newtheorem{theorem}{Theorem}
\newtheorem{proposition}{Proposition}

\newtheorem{corollary}{Corollary}

\DeclareMathOperator*{\argmax}{arg\,max}

\begin{document}

\title{Fundamental Limits of Routing Attack on Network Overload}

\author{Xinyu Wu, Eytan Modiano \\
Laboratory for Information and Decision Systems, MIT, USA \\
\{xinyuwu1,modiano\}@mit.edu}

\maketitle

\IEEEpeerreviewmaketitle


\begin{abstract}
We quantify the threat of network adversaries to inducing \emph{network overload} through \emph{routing attacks}, where a subset of network nodes are hijacked by an adversary. We develop routing attacks on the hijacked nodes for two objectives related to overload: \emph{no-loss throughput minimization} and \emph{loss maximization}. The first objective attempts to identify a routing attack that minimizes the network's throughput that is guaranteed to survive. We develop a polynomial-time algorithm that can output the optimal routing attack in multi-hop networks with global information on the network's topology, and an algorithm with an approximation ratio of $2$ under partial information. The second objective attempts to maximize the throughput loss. We demonstrate that this problem is NP-hard, and develop two approximation algorithms with multiplicative and additive guarantees respectively in single-hop networks. We further investigate the adversary's optimal selection of nodes to hijack that can maximize network overload. We propose a heuristic polynomial-time algorithm to solve this NP-hard problem, and prove its optimality in special cases. We validate the near-optimal performance of the proposed algorithms over a wide range of network settings. Our results demonstrate that the proposed algorithms can accurately quantify the risk of overload given an arbitrary set of hijacked nodes and identify the critical nodes that should be protected against routing attacks.
\end{abstract}

\section{Introduction}


Recent years witnessed the escalating prevalence and severity of network attacks, evidenced by the growing number of incidents and their devastating impact \cite{sermpezis2018survey}. These attacks often result in substantial degradation in network performance, such as lower throughput and higher latency. For example, the 2018 Pakistan Telecom hijacking incident caused considerable network disruption and extensive delays \cite{sun2021securing}, while the 2016 Dyn DDoS attack led to widespread outages that significantly impaired platforms such as Twitter and Netflix \cite{mansfield2016ddos}. 

In this paper, we focus on the degradation of network performance due to \emph{routing attacks}, wherein adversaries hijack  network servers and change their routing decisions \cite{Securing2021}. Routing attacks are a notable form of network attacks with broad impact, which can last for several hours or longer  before being resolved \cite{sermpezis2018survey,al2016bgp}. 
Examples of routing attacks include BGP hijacking, where an attacker falsely claims ownership of an IP prefix to affect routing \cite{schlamp2016heap,cho2019bgp}; routing table poisoning, where false routing information is injected into a victim's routing table \cite{sert2015impacts,song2017novel}; OSPF attacks, which involve fabricating topology information to control routing \cite{nakibly2012persistent}; and blackhole attacks, which diverts the traffic to non-existent destinations \cite{bang2022assessment}. Routing attacks have been detected in a wide range of networks, including data center networks \cite{feldmann2016netco}, software-defined networks \cite{wang2012routing}, and wireless ad hoc networks \cite{sert2015impacts}.

This paper aims to quantify the threat of network adversaries applying routing attacks to
induce \emph{network overload}, where traffic injected into the network cannot be fully transmitted to the destination,
resulting in throughput loss due to network link saturation \cite{li2014dynamic,georgiadis2006optimal,como2012robust}.
A high overload level contributes to performance degradation such as queue instability \cite{shah2011fluid,perry2016chattering} and high latency \cite{venkataramanan2013queue,zhang2022aequitas}, which has become increasingly prevalent in enterprise-level data center networks \cite{li2014dynamic,poutievski2022jupiter,zhang2022aequitas}, due to the demand surge of machine learning workloads \cite{zheng2023traffic,wang2023topoopt} and the slower growth of electrical switch capacity \cite{poutievski2022jupiter, ballani2020sirius}. We develop routing attack strategies for two objectives related to overload: (i) no-loss throughput minimization and (ii) loss maximization. The no-loss throughput of a network represents the maximum traffic arrival rate at the source that can be delivered to the destination without loss. 
The difference between the two objectives is that minimizing the no-loss throughput is equivalent to maximizing the range of traffic arrival rates to the network that will cause overload, while maximizing the loss is equivalent to minimizing the network throughput  given the traffic arrival rate.

Evaluating the routing attacks' ability to inducing network overload remains a relatively unexplored research area \cite{sermpezis2018survey}, unlike other attack types including denial-of-service \cite{fu2019network}, link removal \cite{como2012robust}, and node removal  \cite{tian2017articulation}. 
Routing attacks have presented high risks due to their low implementation cost for adversaries \cite{feldmann2016netco}, inadequate design of detection and defense
mechanisms \cite{sermpezis2018survey}, and broader impacts in software-defined networks where a hijacked controller in the control plane may affect the routing decisions of multiple nodes in the data plane \cite{poutievski2022jupiter,ferguson2021orion}. 
We are motivated to develop algorithms for  network service providers to evaluate the potential impact of network overload that can be caused by routing attacks. We demonstrate that our proposed algorithms can well approximate the maximum level of overload given arbitrary sets of hijacked nodes, and identify the critical nodes that are of highest priority to be protected from routing attacks. 

We summarize the contributions of this work as follows. (i) \emph{No-Loss Throughput Minimization}: We develop an exact polynomial-time routing attack algorithm which can minimize the no-loss throughput in general multi-hop networks with arbitrary sets of hijacked nodes. We further propose a 2-approximation routing attack algorithm purely based on the partial network information downstream to the hijacked nodes. (ii) \emph{Loss Maximization}: We establish the NP-hardness of the loss maximization problem. We develop two approximation algorithms with multiplicative and additive performance guarantee respectively in single-hop networks. (iii) \emph{Optimal Selection of Nodes to Hijack}: We investigate the adversary's optimal selection of nodes to hijack over a set of candidate nodes to optimize the aforementioned two objectives via routing attacks. We prove its NP-hardness in general networks, propose heuristic algorithms and prove the optimality for no-loss throughput minimization when candidate nodes for the adversary are parallel to each other in the network. (iv) \emph{Performance Evaluation}: We evaluate the proposed algorithms and demonstrate their near-optimal performance under a wide range of network settings, including different network densities, default routing policies, and numbers of hijacked nodes. To the best of our knowledge, this is the first quantitative study of routing attacks on network overload with both theory and empirical validation.

The rest of this paper is organized as follows: Section \ref{sec:model_definition} introduces the network models we study and defines the optimal routing attack problem. Section \ref{sec:no_loss_throughput} investigates the no-loss throughput minimization problem. Section \ref{sec:max_loss} investigates the loss maximization problem. Section \ref{sec:optimal_node_selection} investigates the adversary's optimal selection of nodes to hijack and conduct routing attack for the aforementioned two metrics; Section \ref{sec:evaluation} presents the evaluation of the proposed routing attack algorithms under various network settings.

\section{Network Models and Problem Definition}
\label{sec:model_definition}

In this section, we introduce the multi-hop and single-hop network models, and then define the problems of finding the optimal routing attack to minimize no-loss throughput and maximize loss. 

\subsection{Network Models}

\subsubsection{Multi-hop network}

We define a multi-hop network as a graph $\mathcal{G}=(\mathcal{V},\mathcal{E})$, where $\mathcal{V}$ denotes the set of nodes and $\mathcal{E}$ denotes the set of links, with
$|\mathcal{V}| = N$ and $|\mathcal{E}| = M$. We denote the node indices by $1,2,\cdots,N$, and  $(i,j)\in \mathcal{E}$ if there is a link from node $i$ to $j$. We consider a single commodity with node $1$ as the source and node $N$ as the destination. Denote the external traffic arrival rate to the network at the source by $\lambda$, and the traffic transmission rate over link $(i,j)$ by $f_{ij}$, defined as the amount of traffic over $(i,j)$ in a time unit.
Each link has a capacity value, denoted by $c_{ij}$ which is the maximum transmission rate over $(i,j)$, i.e., $f_{ij}\in [0, c_{ij}]$. {We treat $(i,j)$ and $(j,i)$ as different links}. 

The network traffic may have multiple available paths from node $1$ to $N$, where each node dispatches traffic to its connected nodes. We define the dispatch ratio vector at a non-destination node $i$ as a vector $\mathbf{x}_i \in \mathbb{R}^N$, where each element $x_{ij}$ denotes the fraction of traffic sent from node $i$ to its connected node $j$ among the total traffic at node $i$, with $x_{ij} = 0$ for $(i,j)\notin \mathcal{E}$, and $\sum_{j=1}^N x_{ij} = 1$, where the sum of dispatch ratios of node $i$ should be $1$. 
We call $\mathbf{x}_i$ the \emph{routing policy} at node $i$, and $\mathbf{X} = [\mathbf{x}_i]_{i=1,\cdots,N}$ denotes the \emph{routing matrix} of the network\footnote{$\mathbf{x}_N = \boldsymbol{0}$. We put node $N$ in to keep the dimension equal to network size.}. 

Examples of the routing policies defined above include random packet spraying \cite{dixit2013impact}, where a node dispatches traffic uniformly to its connected nodes, i.e., $x_{ij} = \frac{1}{|k\in \mathcal{V} \mid (i,k)\in \mathcal{E}|},~\forall (i,j)\in \mathcal{E}$, and weighted spraying based on link capacity where $x_{ij} = \frac{c_{ij}}{\sum_{k: (i,k)\in \mathcal{E}} c_{ik}},~\forall (i,j)\in \mathcal{E}$. 
We further show in Proposition \ref{prop:equivalence} the equivalence between the above dispatch ratio characterization and the multi-path characterization for a commodity  \cite{fu2019network}, for example the shortest-path and equal-cost multi-path routing \cite{zhou2014wcmp}. 
Henceforth, we will solely use the dispatch ratio vector to model the routing policy. We give an example of a routing policy in Fig.~\ref{fig:Paper_example_multihop_networks_and_routing_attack}(a) where at node $1$, $(x_{12},x_{13})=(0.5,0.5)$.

\begin{proposition}
\label{prop:equivalence}
{Suppose there are $P$ paths 
$\{Path_p\}_{p=1}^P$ of a commodity, 
where a fraction $\theta_p$ of the traffic takes $Path_p$. The corresponding routing matrix $\mathbf{X}$ satisfies $x_{ij} = \beta_{ij} / \sum_{k:(i,k)\in \mathcal{E}}\beta_{ik},~\forall (i,j)\in \mathcal{E}$, where $\beta_{ij} = \sum_{p: (i,j)\in Path_p} \theta_p$, and $x_{ij}=0,~\forall (i,j)\notin \mathcal{E}$, where  $(i,j) \in Path_p$ if link $(i,j)$ is on $Path_p$.} 
\end{proposition}

\begin{proof}
We denote $i \overset{p}{\rightarrow} j$ for a path $Path_p$ if node $j$ is the next node to be visited after node $i$ in a path $Path_p$. Given the available paths, the total fraction of traffic that goes through link $(i,j)$ is $\beta_{ij} := \sum_{p: i,j \in Path_p, i \overset{p}{\rightarrow} j} \theta_p$, and total fraction of traffic that takes paths containing node $i$ is $\beta_{i} := \sum_{j:(i,j)\in \mathcal{E}}\beta_{ij}$. Then we can construct the dispatch ratio $x_{ij} = \beta_{ij} / \beta_{i}$ for any $(i,j)\in \mathcal{E}$. Meanwhile, given the routing matrix, the portion of traffic that takes $Path_p = \{s,i_1^p, \cdots, i_L^p, d\}$ for some $L>0$ can be calculated as $\theta_p = x_{s,i_1^p} \times \prod_{l=2}^L x_{i_{l-1}^p,i_{l}^p} \times x_{i_L^p,d}$ ($L=0$ is the special case where $Path_p = \{s,d\}$ where $\theta_p = x_{s,d}$).
\end{proof}

We can characterize the \emph{overload} at a link $(i,j)$ using the definition of the routing policy. If the traffic to be sent through $(i,j)$ exceeds $c_{ij}$, i.e., $\left(\sum_{k: (k,i)\in \mathcal{E}} f_{ki}\right) x_{ij} > c_{ij}$, then link $(i,j)$ will be \emph{saturated}, i.e., $f_{ij}=c_{ij}$, and we say that overload occurs at $(i,j)$, where the excess traffic $\left(\sum_{k: (k,i)\in \mathcal{E}} f_{ki}\right) x_{ij}-c_{ij}$ will be dropped, and lost.

We define the \emph{upstream} and \emph{downstream} node set used in algorithm design. The upstream node set of a node $i$, denoted by $\mathcal{V}_i^{up}$, contains all the nodes except $i$ that have a path in the available path set to transmit traffic to $i$ under the routing matrix $\mathbf{X}$. 
The downstream node set of a node $i$, denoted by $\mathcal{V}_i^{down}$, contains all the nodes except $i$ to which there exists a path in the available path set starting from $i$. We define a node $j\in \mathcal{V}_i^{down}$ to be a connected downstream node of $i$ if further $(i,j)\in \mathcal{E}$. {We define a link $(j,k)$ to be a downstream link of node $i$ if $j,k \in \mathcal{V}_i^{down} \cup \{i\}$, and a connected downstream link of node $i$ if further $j = i$.}
{Note that $\mathcal{V}_i^{up} \cap \mathcal{V}_i^{down} = \emptyset$ holds in directed acyclic networks, while there may exist nodes that are simultaneously the upstream and downstream node of a node $i$ in general networks. The framework and analysis in this paper hold for general networks. However, for ease of understanding we use examples with $\mathcal{V}_i^{up} \cap \mathcal{V}_i^{down} = \emptyset$ for explanation. For example in Fig.~\ref{fig:Paper_example_multihop_networks_and_routing_attack}(a), we have $\mathcal{V}_5^{up} = \{1,3\}$ and $\mathcal{V}_2^{down}=\{4,6\}$.}


\subsubsection{Single-hop network} A single-hop network contains a set of ingress nodes and egress nodes. Examples include server farms, switched networks, and the basic structure of data center networks like Fat-Tree \cite{al2008scalable} and Clos \cite{singh2015jupiter}, as shown in Fig.~\ref{fig:Paper_Single_Hop_Network_Examples}(a) and (b). We develop theoretical performance guarantees for our proposed attack algorithms under single-hop networks in Section \ref{sec:max_loss}. 

\begin{figure}[!htbp]
\centering
\includegraphics[width=1.0\linewidth]{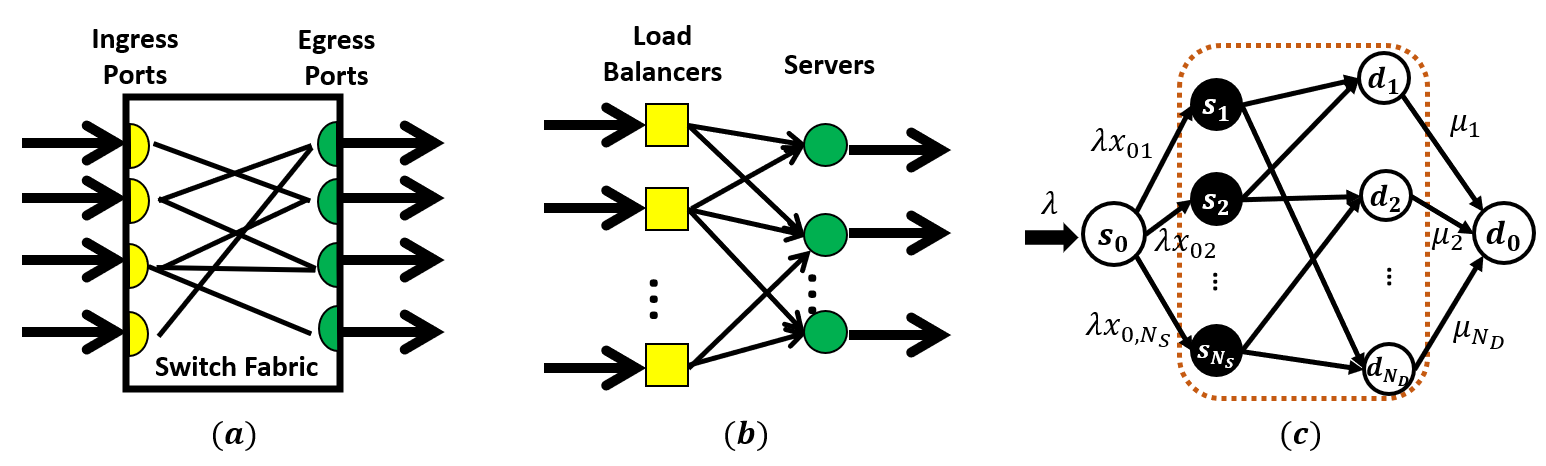}
\caption{(a) Switch network; (b) Server farm; (c) Bipartite graph (in the dashed box) with meta source $s_0$ and destination $d_0$}
\label{fig:Paper_Single_Hop_Network_Examples}
\end{figure}

An $N_S \times N_D$ single-hop network
is a bipartite graph $\mathcal{G}=(\mathcal{V},\mathcal{E})$ with $\mathcal{V}:=\{\mathcal{V}_S, \mathcal{V}_D\}$, where $\mathcal{V}_S$ and $\mathcal{V}_D$ represent the set of ingress and egress nodes respectively, and $\mathcal{E}$ denotes the set of links between $\mathcal{V}_S$ and $\mathcal{V}_D$, and $|\mathcal{V}_S|=N_S$ and $|\mathcal{V}_D|=N_D$. 
Denote the $i$th ingress node by $s_i$ and $j$th egress node by $d_j$. 
We use slightly different notations in single-hop networks as follows. Denote the traffic arrival rate at ingress node $s_i$ by $\lambda_i$, the service rate of egress node $d_j$ by $\mu_j$, and the transmission rate over link $(s_i,d_j)$ by $f_{ij}$, with their corresponding vector forms being $\boldsymbol{\lambda}:=\{\lambda_i\}_{i=1}^{N_S}$, $\boldsymbol{\mu}:=\{\mu_j\}_{j=1}^{N_D}$, and $\mathbf{f}=\{f_{ij}\}_{(i,j)\in \mathcal{E}}$, respectively. Denote the routing policy at $s_i$ by $\mathbf{x}_i = \{x_{ij}\}_{d_j\in \mathcal{V}_D}$.
We consider sufficient capacity over links between $\mathcal{V}_S$ and $\mathcal{V}_D$ so that they will not be saturated, with a primary focus on how a routing attack at $\mathcal{V}_S$ affects overload at $\mathcal{V}_D$ \cite{li2014dynamic,fu2019fundamental}. Overload occurs at $d_j\in \mathcal{V}_D$ if $\sum_{i=1}^{N_S} \lambda_i x_{ij} > \mu_j$ under routing matrix $\mathbf{X}$.

We can transform the single-hop structure to a commodity starting from a meta-source $s_0$, connected to each ingress node in $\mathcal{V}_S$, to a meta-destination $d_0$ that receives traffic from each egress node in $\mathcal{V}_D$, by setting the traffic arrival rate to $s_0$ to be $\lambda = \sum_{i=1}^{N_S} \lambda_i$, the dispatch ratio from $s_0$ to $s_i$ to be $x_{0i} = \lambda_i / \lambda$, and the capacity of link $(d_j,d_0)$ to be $\mu_j$. It is straightforward to see that the transformed graph is equivalent to the bipartite graph, shown in Fig.~\ref{fig:Paper_Single_Hop_Network_Examples}(c).

\textbf{Remark:} We focus on static parameters in this paper. We anticipate that the methodology we propose will inspire addressing time-varying network parameters in future work.



\subsection{Problem Definition}

We consider a network adversary who hijacks and gains access to  the routing policies of certain network nodes
$\mathcal{V}_A \subseteq \mathcal{V}$. We call $\mathcal{V}_A$ the set of  \emph{adversarial nodes}, and $\mathcal{V}_N:=\mathcal{V}\backslash 
\mathcal{V}_A$ the set of \emph{normal nodes}. Each normal node $i \in \mathcal{V}_N$ follows some \emph{default} routing policies $\mathbf{x}_i$, and the adversary aims to find the routing policies at adversarial nodes, denoted by $\mathbf{x}_A := \{\mathbf{x}_i\}_{i\in \mathcal{V}_A}$, to maximize overload. 
In the example in Fig.~\ref{fig:Paper_example_multihop_networks_and_routing_attack}(a), $\mathcal{V}_A = \{3\}$ and $\mathcal{V}_N=\{1,2,4,5,6\}$. The whole numbers denote link capacities and the highlighted fraction over a link denotes the dispatch ratio. We use shaded nodes to represent the adversarial nodes and unshaded nodes to represent the normal nodes.


We optimize routing attacks for two objectives related to overload given $\mathcal{V}_A$. (i) \textbf{Minimize No-Loss Throughput}: We define \emph{no-loss throughput}, denoted by $\lambda^*$, as the maximum traffic arrival rate $\lambda$ at the source node that will not lead to link overload given the routing matrix $\mathbf{X}$. Equivalent interpretations of $\lambda^{*}$ include the maximum traffic arrival rate that can be successfully transmitted \cite{como2012robust2}, the max arrival rate that guarantees maximum link utilization below $100\%$ \cite{zhang2021gemini}, and the max arrival rate that ensures queue stability \cite{fu2019fundamental}. 
Essentially, $\lambda^*$ captures the network robustness to routing attacks.
We investigate the capability of routing attack over $\mathcal{V}_A$ to minimize $\lambda^*$, i.e., minimizing network robustness to overload. {The minimum $\lambda^*$, denoted by $\lambda_{OPT}^*$ in this paper, can be interpreted as the upper limit of traffic arrival rate that guarantees no traffic loss under any adversarial routing attack on the given $\mathcal{V}_A$}. (ii) \textbf{Maximize Loss}: Given the default routing policies at normal nodes $\mathcal{V}_N$ and the arrival rate $\lambda$, the adversary manipulates the routing policies in $\mathcal{V}_A$ to maximize the throughput loss caused by link overload, equivalent to minimizing the total amount of traffic that can be transmitted to the destination. Loss maximization reflects the most severe overload that can be caused via a routing attack. These two objectives reflect different facets of network overload, and we show in Fig.~\ref{fig:Paper_example_multihop_networks_and_routing_attack}(b) and \ref{fig:Paper_example_multihop_networks_and_routing_attack}(c) that their optimal routing attack policies are different, which motivates different algorithm design patterns in Section \ref{sec:no_loss_throughput} and \ref{sec:max_loss}. 

In summary, we give the formal statement of the core problem of this paper: \emph{Given $\mathcal{G}$, capacity $\{c_{ij}\}_{(i,j)\in \mathcal{E}}$, default routing policies at normal nodes $\{\mathbf{x}_i\}_{i\in \mathcal{V}_N}$, and adversarial nodes $\mathcal{V}_A$, what is the optimal routing attack policy over $\mathcal{V}_A$ to minimize no-loss throughput and maximize loss?} 

\begin{figure}[!htbp]
\centering
\includegraphics[width=0.99\linewidth]{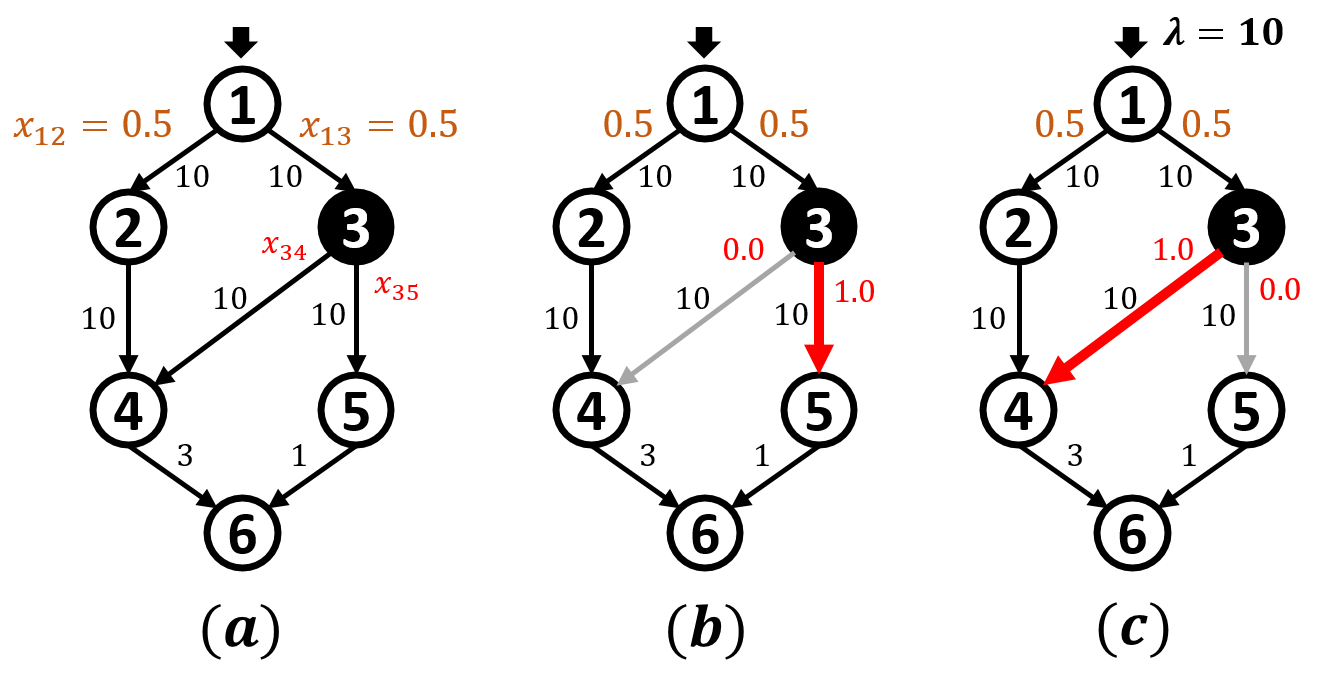}
\caption{\small (a) A 6-node network with $\mathcal{V}_A=\{3\}$ and $\mathcal{V}_N = \{1,2,4,5,6\}$; (b) Optimal routing to minimize $\lambda^*$ is $(x_{34}, x_{35}) = (0,1)$: $\lambda_{OPT}^*$ is $2$ and $(5,6)$ is the first saturated link; (c) Given $\lambda=10$, the optimal routing to maximize loss is $(x_{34}, x_{35}) = (1,0)$, with maximum loss of $\lambda - c_{46} = 10 - 3 = 7$.}
\label{fig:Paper_example_multihop_networks_and_routing_attack}
\end{figure}

Furthermore, in Section \ref{sec:optimal_node_selection} we investigate the \emph{optimal node selection} problem: The adversary needs to identify the optimal nodes to hijack from a set of candidate nodes to conduct routing attack in order to minimize no-loss throughput or maximize loss. This problem fits into the practice where the adversary can only hijack and manipulate the routing policies over a limited number of nodes due to the cost or the risk of being exposed. The answer to this problem informs network service providers of the critical nodes that should be protected against routing attacks whose control by the adversary can lead to high risk of overload and traffic loss.


\section{No-Loss Throughput Minimization}
\label{sec:no_loss_throughput}

In this section, we investigate the optimal routing attack to minimize no-loss throughput $\lambda^*$. We develop an exact polynomial-time algorithm that returns an optimal routing attack based on linear programming (LP), given the network topology, link capacities, and default routing policies at $\mathcal{V}_N$. We then study the case where the adversary only has access to the network information downstream to the adversarial nodes $\mathcal{V}_A$, and propose a routing attack algorithm with an approximation ratio of at most $2$. 
The key takeaway is that the adversary can execute routing attacks efficiently to achieve optimal performance in diminishing the network's robustness against overload in general network settings given arbitrary sets of $\mathcal{V}_A$, which quantitatively unveils the threat of routing attacks in increasing the risk of overload.


\subsection{Problem Formulation}

We formulate the $\lambda^*$-minimization problem as follows.

\begin{equation}
\label{eqn:min-max-lambda-no-loss}
\begin{aligned}
\lambda_{OPT}^* := ~ &\min\limits_{\mathbf{f}} \quad  \max\limits_{\lambda} \quad \lambda \\
\textrm{s.t.} \quad & \mathbf{f} \in \Lambda, \quad
\lambda = f_{01}, \quad  f_{ij} \in [0, c_{ij}], \forall (i,j) \in \mathcal{E}
\end{aligned}
\end{equation}
where $\Lambda$ includes both $\Lambda_N$ and $\Lambda_A$ below.
\begin{equation}
\small
\label{eqn:constraints}
\begin{cases}
\Lambda_N: f_{ij} = \left(\sum_{k:(k,i)\in \mathcal{E}} f_{ki}\right)x_{ij}, \forall i \in \mathcal{V}_N,~\forall (i,j) \in \mathcal{E} \\
\Lambda_A: \sum_{j:(i,j)\in \mathcal{E}}f_{ij}=\sum_{k:(k,i)\in \mathcal{E}}f_{ki},\forall i\in \mathcal{V}_A
\end{cases}.
\end{equation}
 
In \eqref{eqn:min-max-lambda-no-loss}, we introduce a meta source node 0 that serves traffic only to the original source node 1 with $f_{01}=\lambda$ and $c_{01}\geq \lambda$, so that the constraints at node 1 can be included in \eqref{eqn:constraints}. The decision variables of \eqref{eqn:min-max-lambda-no-loss} are the flow variables $\mathbf{f}:=\{f_{ij}\}_{(i,j)\in \mathcal{E}}$. In constraints \eqref{eqn:constraints}, $\Lambda_N$ represents the flow conservation at any normal node $i \in \mathcal{V}_N$, where the transmission rate over link $(i,j)$ is equal to the total traffic injection to node $i$ multiplied by the dispatch ratio $x_{ij}$ from $i$ to $j$. $\Lambda_A$ represents the flow conservation at any adversarial node $i \in \mathcal{V}_A$, where the adversary can manipulate routing arbitrarily subject to flow conservation. The optimal solution to \eqref{eqn:min-max-lambda-no-loss}, denoted by $\mathbf{f}^*$, leads to the optimal routing attack, denoted by $\mathbf{x}_A^*:=\{\mathbf{x}_i^*\}_{i\in \mathcal{V}_A}$, where $x_{ij}^* = f_{ij}^* / \sum_{k:(k,i)\in \mathcal{E}}f_{ki}^*$ for $(i,j)\in \mathcal{E}$. Using $\mathbf{f}$ instead of $\mathbf{x}_A$ as decision variables can render the constraints to a linear form, however \eqref{eqn:min-max-lambda-no-loss} still cannot be formulated as a standard LP problem. We transform \eqref{eqn:min-max-lambda-no-loss} into the following equivalent optimization framework \eqref{eqn:min-max-lambda-no-loss-equi} which we show later can pave the way for the polynomial-time algorithm designn. We prove the equivalence of this transformation in Proposition \ref{prop:equi-transform}.

\begin{equation}
\label{eqn:min-max-lambda-no-loss-equi}
\begin{aligned}
\max_{\mathbf{f}} \quad & \max_{(i,j)\in \mathcal{E}}   \quad f_{ij} / c_{ij} \\
\textrm{s.t.} \quad
& \mathbf{f} \in \Lambda, \quad 
f_{01} = 1, \quad f_{ij}\geq 0, ~\forall (i,j)\in \mathcal{E}.
\end{aligned}
\end{equation}

\begin{proposition}
\label{prop:equi-transform}
{The optimal routing policies given by the optimal solutions $\mathbf{f}^*$ to \eqref{eqn:min-max-lambda-no-loss} and \eqref{eqn:min-max-lambda-no-loss-equi} are equivalent.} The minimum no-loss throughput $\lambda_{OPT}^{*}$ is given by $(\max_{(i,j)\in \mathcal{E}} f_{ij}^*/c_{ij})^{-1}$ from \eqref{eqn:min-max-lambda-no-loss-equi}.
\end{proposition}

\begin{proof}

In \eqref{eqn:min-max-lambda-no-loss}, define new variables 
$\tilde{f}_{ij}:=f_{ij}/\lambda, \forall (i,j)\in \mathcal{E}$. Then the optimization in \eqref{eqn:min-max-lambda-no-loss} is equivalent to 

\begin{equation}
\begin{aligned}
\min_{\mathbf{\tilde{f}}} \quad & \max\limits_{\lambda} \quad \lambda \\
\textrm{s.t.} \quad
& \tilde{\mathbf{f}} \in \Lambda, \quad 
\tilde{f}_{01} = 1, \quad \tilde{f}_{ij} \in 
[0, c_{ij}/\lambda], ~
\forall (i,j) \in \mathcal{E}.
\end{aligned}
\end{equation}

We have $\lambda \leq c_{ij}/\tilde{f}_{ij}$ for each link, and thus $\lambda \leq \min_{(i,j)\in \mathcal{E}} c_{ij}/\tilde{f}_{ij}$. Note that this is the only constraint for $\lambda$, and $\lambda$ is an independent decision variable after the transformation. Therefore the maximum $\lambda$ that guarantees no loss is equal to $\min_{(i,j)\in \mathcal{E}} c_{ij}/\tilde{f}_{ij}$, and the objective function becomes 
$\min_{\tilde{\mathbf{f}}} \min_{(i,j)\in \mathcal{E}} {c_{ij}}/{\tilde{f}_{ij}} = \max_{\tilde{\mathbf{f}}} \max_{(i,j)\in \mathcal{E}} {\tilde{f}_{ij}}/{c_{ij}}$
which is exactly \eqref{eqn:min-max-lambda-no-loss-equi} by changing notation  $\tilde{f}_{ij}$ to $f_{ij}$. 
\end{proof}

The formulation \eqref{eqn:min-max-lambda-no-loss-equi} lets $f_{01}=1$, i.e., one unit of traffic arrival rate to the network\footnote{We can choose arbitrary constant for $f_{01}$ since we evaluate the saturation level and thus there is no constraint on $f_{ij}\in [0, c_{ij}]$.}, and the goal is to find the optimal routing policy that maximizes the maximum link utilization, which is $\max_{(i,j)\in \mathcal{E}} f_{ij}/c_{ij}$. Then the minimum possible $\lambda^*$ under an optimal routing attack is $(\max_{(i,j)\in \mathcal{E}} f_{ij}^*/c_{ij})^{-1}$ where $\mathbf{f}^*$ is the optimal solution to \eqref{eqn:min-max-lambda-no-loss-equi}. The intuition is that the link with the highest utilization is the first saturated link when the arrival rate increases gradually from $0$. 
The critical point of Proposition \ref{prop:equi-transform} is that the new link-wise maximization formulation \eqref{eqn:min-max-lambda-no-loss-equi} inspires the following design of polynomial-time algorithms. 

\subsection{Exact Algorithms}
\label{sec:polynomial-algorithm}

We propose exact algorithms to minimize no-loss throughput $\lambda^*$ in polynomial time. We derive a simple brute-force algorithm which has polynomial time complexity under $|\mathcal{V}_A|=O(1)$, and a polynomial-time algorithm under general $\mathcal{V}_A$ based on solving an LP for maximum flow.

\subsubsection{Brute-Force under $|\mathcal{V}_A|=O(1)$}
We prove Theorem \ref{thm:boundary_property} below, which states that there must exist an optimal routing policy under which each adversarial node $i\in \mathcal{V}_A$ dispatches all traffic flows to a single downstream connected node. 


\begin{theorem}
\label{thm:boundary_property}
There exists an optimal solution $\mathbf{f}^*$ to \eqref{eqn:min-max-lambda-no-loss-equi} such that for $\forall i \in \mathcal{V}_A$, $\exists j$ that $(i,j) \in \mathcal{E}$ and $x_{ij}^*=\frac{f_{ij}^*}{\sum_{k:(i,k)\in \mathcal{E}}f_{ik}^*}=1$, while $x_{ik}^*=0$ for $k\neq j$.
\end{theorem}

\begin{proof}
The objective function of problem \eqref{eqn:min-max-lambda-no-loss-equi} is convex with respect to the flow variables $\mathbf{f}=\{f_{ij}\}_{(i,j)\in \mathcal{E}}$ due to the convexity of the elementise maximum function. Meanwhile the constraints are all linear which form a polytope. Therefore \eqref{eqn:min-max-lambda-no-loss-equi} is a convex maximization problem, where at least one of the optimal solutions is at one of the vertices of the polytope. 
There is a total of $M+1$ decision variables (including $f_{01}$). Denote the number of links starting from $\mathcal{V}_N \cup \{0\}$ as $M_N$, and links starting from $\mathcal{V}_A$ as $M_A$. Then the total number of equality constraints in \eqref{eqn:min-max-lambda-no-loss-equi} is $M_N + |\mathcal{V}_A|$, and thus at any vertex there should have at least $M+1 - M_N - |\mathcal{V}_A| = M_A - |\mathcal{V}_A|$ flow variables $f_{ij}$ over link $(i,j)$ to be zero. Each $i \in \mathcal{V}_A$ routing traffic to a single connected downstream  node guarantees the above condition, i.e., $\exists j$ that $(i,j) \in \mathcal{E}$ and $x_{ij}^*=\frac{f_{ij}^*}{\sum_{k:(i,k)\in \mathcal{E}}f_{ik}^*}=1$, while $x_{ik}^*=0$ for $k\neq j$.
\end{proof}

The intuition of Theorem \ref{thm:boundary_property} is that \eqref{eqn:min-max-lambda-no-loss-equi} maximizes a convex function over a polytope, hence one of the vertices of the polytope must be the optimal solution.
The implication of Theorem \ref{thm:boundary_property} is that the optimal attack to minimize $\lambda^*$ can be identified in a brute force manner by exhausting all the combinations where each adversarial node sends all the traffic to one of its connected downstream nodes. {The upper bound on the number of combinations is $d_{\max}^{|V_A|}$, where $d_{\max}$ denotes the maximum number of connected downstream links of a node. Note that $d_{\max}=O(N)$ and thus $d_{\max}^{|V_A|}$ is a polynomial function of network size $N$ when $|V_A| = O(1)$.}
We give an example with $|\mathcal{V}_A|=2$ in Fig.~\ref{fig:Paper_example_boundary_property}.

\begin{figure}[!htbp]
\centering
\includegraphics[width=0.95\linewidth]{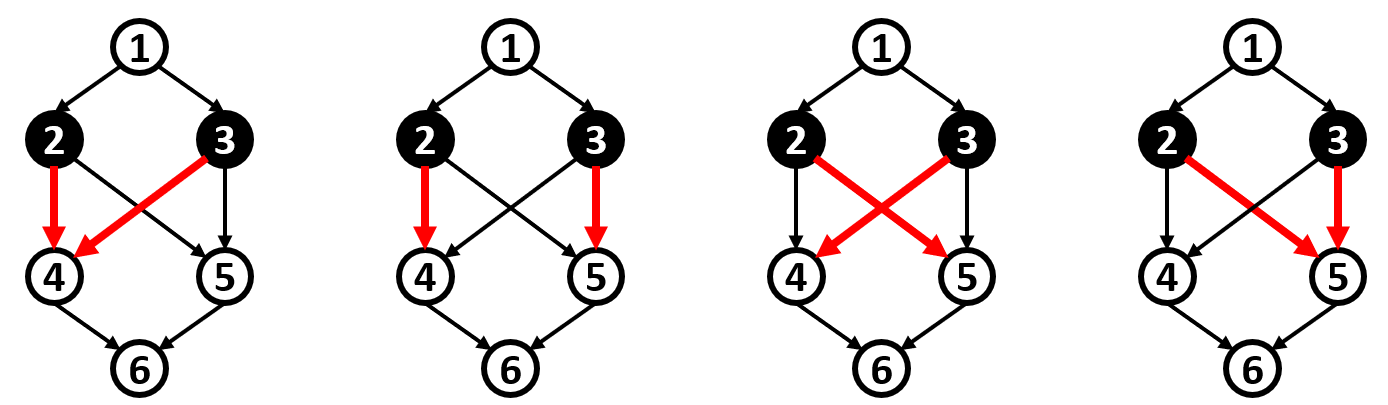}
\caption{One of the 4 combinations must be optimal, where $\mathcal{V}_A=\{2,3\}$ serve all traffic through the highlighted links.}
\label{fig:Paper_example_boundary_property}
\end{figure}

For each of the above combinations of routing attacks, we can calculate $\lambda^*$ given the routing matrix $\mathbf{X}$ by a two-step process: First, for each $(i,j)\in \mathcal{E}$, calculate the proportion of the traffic starting from the source node that will go through this link, assuming that link capacities are infinite. $f_{ij}$ given $\lambda = 1$ under this setting represents the proportion of traffic that goes over link $(i,j)\in \mathcal{E}$. In general, this can be done by solving the linear equations built via the first constraint in \eqref{eqn:constraints} given the routing policy at each node $i$, with worst time complexity $O(M^3)$. For directed acyclic networks, the time complexity can be reduced to $O(M)$ by assuming $\lambda=1$ and then following $f_{ij} = \left(\sum_{k:(k,i)\in \mathcal{E}} f_{ki}\right)x_{ij}$ over all the $M$ links. Secondly, resume the original link capacities, and find the link $(i^*, j^*) = \arg\max_{(i,j)\in \mathcal{E}} f_{ij}/c_{ij}$ in $O(M)$ time. Then $\lambda^* = (f_{i^*j^*} / c_{i^*j^*})^{-1}$. The limitation of the algorithm is the exponential time complexity under $\mathcal{V}_A$ with $|\mathcal{V}_A| = O(N)$.



\subsubsection{MaxFlow-based Solution for General $\mathcal{V}_A$} 
We further develop a polynomial-time algorithm that outputs the optimal routing attack to \eqref{eqn:min-max-lambda-no-loss-equi} under general $\mathcal{V}_A$, presented in Algorithm \ref{Alg:optimal-metric-1-general}.
First, we figure out the routing policy at $\mathcal{V}_A \cap \mathcal{V}_i^{up}$ that maximizes the total flow to each node $i$ (except destination), assuming arrival rate $\lambda = 1$ at the source and infinite capacity over all the links. Denote the corresponding maximum possible flow to node $i$ by $MF[i]$. $MF[1]=1$ trivially, while for $i>1$, $MF[i]$ can be obtained by solving an LP that outputs the maximum flow solution. Second, for each non-destination node $i$, we identify the first link that can be saturated under routing attack among all the connected downstream links of node $i$ as follows. If $i\in \mathcal{V}_N$, the first saturated link starting from $i$ should be $(i,j_i^*)=\argmax_{(i,j)\in \mathcal{E}} x_{ij} / c_{ij}$; If $i\in \mathcal{V}_A$, the adversary can arbitrarily adjust the routing at $i$, thus the first-saturated link starting from $i$ is $(i,j_i^*)=\arg\min_{(i,j)\in \mathcal{E}} c_{ij}$. 
It is straightforward to see that the first link that will be saturated under the optimal routing attack given by \eqref{eqn:min-max-lambda-no-loss-equi} must be among the $N-1$ first-saturated links $\{(i,j_i^*)\}_{i=1}^{N-1}$.
Combining the two steps as done in Algorithm \ref{Alg:optimal-metric-1-general} establishes its correctness to solving \eqref{eqn:min-max-lambda-no-loss-equi}, which outputs the routing attack solution with minimum traffic arrival rate required to saturate one of the $N-1$ first-saturated links, and the resulting minimum arrival rate is $\lambda_{OPT}^*$ in \eqref{eqn:min-max-lambda-no-loss-equi}. We state this result in Theorem \ref{thm:polynomial_max_flow}. We give a toy example running Algorithm \ref{Alg:optimal-metric-1-general} in Fig.~\ref{fig:Paper_example_polynomial_algorithm}. 

\begin{theorem}
\label{thm:polynomial_max_flow}
    Algorithm \ref{Alg:optimal-metric-1-general} outputs the optimal solution to \eqref{eqn:min-max-lambda-no-loss-equi}.
\end{theorem}

\begin{algorithm}[!htp]
\caption{Exact Algorithm to Minimize $\lambda^*$ for General $\mathcal{G}=(\mathcal{V}, \mathcal{E})$ and $\mathcal{V}_A$}
\label{Alg:optimal-metric-1-general}
\textbf{Input:} $\mathcal{G}=(\mathcal{V}, \mathcal{E})$, $\mathcal{V}_A$, default routing $\{\mathbf{x}_i\}_{i\in \mathcal{V}_N}$\;
\For{$\forall i \in \mathcal{V}\backslash\{N\}$}{
Given one unit of arrival while assuming unlimited capacity, calculate $MF[i]$, the max flow to node $i$\;
\lIf{$i\in \mathcal{V}_N$}{$V[i]=MF[i]\times x_{ij^*}/c_{ij^*}$ with $(i,j^*)=\argmax_{(i,j)\in \mathcal{E}} x_{ij} / c_{ij}$}
\lElse{$V[i]=MF[i]/c_{ij^*}$ with $(i,j^*)=\arg\min_{(i,j)\in \mathcal{E}} c_{ij}$}
}
Find $i^* = \arg\max_{i\in \mathcal{V}} V[i]$\;
Construct the routing attack at $\forall i \in \mathcal{V}_A \cap (\mathcal{V}_{i^*}^{up} \cup \{i^*\})$, corresponding to node $i^*$, denoted by $\mathbf{x}_i^{(i^*)}$\;
For all $i\in \mathcal{V}_A$, let $\mathbf{x}_i = \mathbf{x}_i^{(i^*)}$ if $i\in \mathcal{V}_{i^*}^{up} \cup \{i^*\}$ else arbitrarily set the routing policy $\mathbf{x}_i$\;
\textbf{Return} $\mathbf{x}_A=\{\mathbf{x}_i\}_{i\in \mathcal{V}_A}$\;
\end{algorithm}

\begin{figure}[!htbp]
\centering
\includegraphics[width=0.99\linewidth]{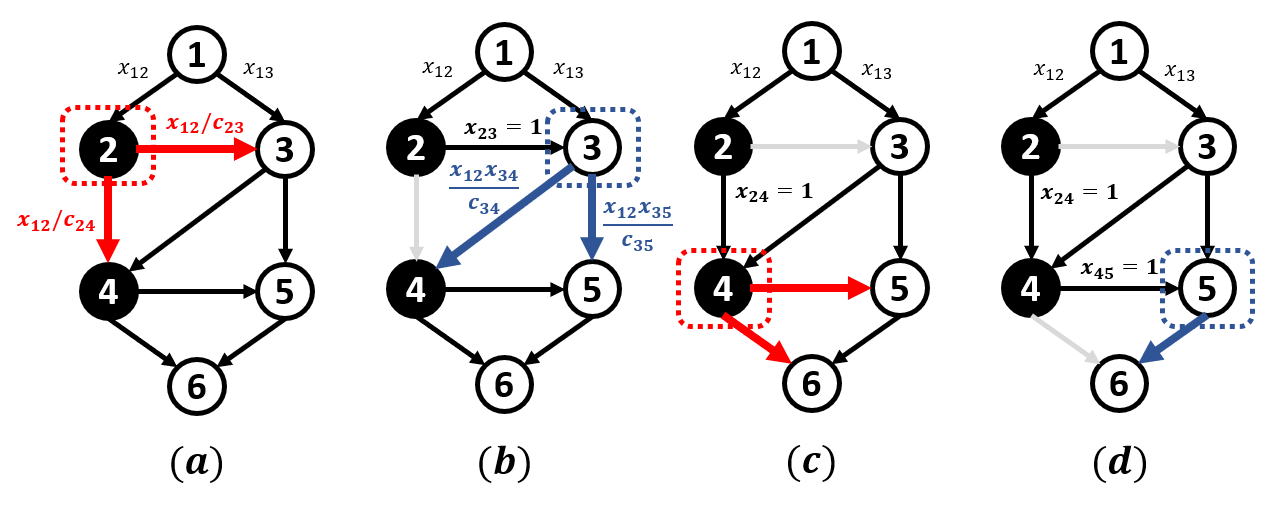}
\caption{\small Example of Algorithm \ref{Alg:optimal-metric-1-general}. Assume that $c_{12}, c_{13}\rightarrow \infty$ which means $(1,2)$ and $(1,3)$ will not be saturated. We can calculate $MF[2] = x_{12}$, $MF[3] = 1$ since the adversarial node $2$ can route all packets to $3$, $MF[4] = x_{12} + x_{13}x_{34}$ since node $2$ can route all packets to $4$, and $MF[5] = 1$ since node $4$ can route all packets to $5$. Then we find the first saturated connected downstream link of each node in $\{2,3,4,5\}$, where the calculation in (a) and (c) follows line 5 in Algorithm \ref{Alg:optimal-metric-1-general} since $2,4\in \mathcal{V}_A$ with links highlighted in red,  while that in (b) and (d) follows line 4 since $3,5\in \mathcal{V}_N$ highlighted in blue.}
\label{fig:Paper_example_polynomial_algorithm}
\end{figure}

The computation cost of Algorithm \ref{Alg:optimal-metric-1-general} mainly lies in solving LPs to obtain $MF[i]$ for each non-destination node $i$. We can apply interior point methods to solve LPs in polynomial time. 

\subsection{2-Approximation Algorithm with Partial Information}

We extend to the situation where each adversarial node $i\in \mathcal{V}_A$ only has the access to the network information of its downstream nodes in $\mathcal{V}_i^{down}$. This is realistic as switches and routers may only store the routing information downstream. We demonstrate that even with such partial information, there exists a $2$-approximation algorithm for minimizing $\lambda^*$, which means $\lambda_{ALG}^*$, the no-loss throughput under the routing attack given by the algorithm, satisfies $\lambda_{ALG}^* \leq 2\lambda_{OPT}^*$, where $\lambda_{OPT}^*$ is the no-loss throughput under the optimal solution to \eqref{eqn:min-max-lambda-no-loss}. It unveils the potential of adversaries to find a routing attack close to the optimal attack with partial network information.

We propose the approximation algorithm in Algorithm \ref{Alg:approximation-2}. The core idea is to decompose the traffic flows into two sets: flows that pass some $i\in \mathcal{V}_A$, and those that do not pass any adversarial node. It constructs the routing attack based on the first set of flows, which we show can be derived purely based on the topology and routing information downstream to $\mathcal{V}_A$. 
Algorithm \ref{Alg:approximation-2} contains three steps: (i) Construct the downstream subgraph starting from  $\mathcal{V}_A$, denoted by $\mathcal{G}_{A}^{down} = (\mathcal{V}_{A}^{down}, \mathcal{E}_{A}^{down})$, where $\mathcal{V}_{A}^{down} = \mathcal{V}_A \cup \{\mathcal{V}_i^{down}\}_{i\in \mathcal{V}_A}$ and $\mathcal{E}_{A}^{down} = \{(i,j)\in \mathcal{E}\mid i,j \in \mathcal{V}_{A}^{down}\}$; (ii) For each $i\in \mathcal{V}_A$, calculate the ratio of traffic flows that are transmitted from the source node to node $i$ without passing any other adversarial node $j \in \mathcal{V}_A \backslash \{i\}$, denoted by $R[i]$, based on the downstream traffic information introduced later in Algorithm \ref{Alg:calculation_Ri}; (iii) Find the optimal attack in $\mathcal{G}_A^{down}$ by Algorithm \ref{Alg:optimal-metric-1-general}, considering $R[i]$ units of arrival to each $i\in \mathcal{V}_A$.

\begin{algorithm}[!htp]
\caption{2-Approximation algorithm to minimize $\lambda^*$ based on $\mathcal{G}_A^{down}$}
\label{Alg:approximation-2}
\textbf{Input:} $\mathcal{V}_A$, $\mathcal{G}_{A}^{down}$, normal routing $\{\mathbf{x}_{i}\}_{i\in \mathcal{V}_N \cap \mathcal{V}_A^{down}}$\;
Determine $R[i], \forall i\in \mathcal{V}_A$ by Algorithm \ref{Alg:calculation_Ri}\;
Add $R[i]$ units of arrival rate to node $i\in \mathcal{V}_A$ and solve the optimal solution $\{\mathbf{x}_i^*\}_{i\in \mathcal{V}_A}$\ by \eqref{eqn:min-max-lambda-no-loss-equi} on $\mathcal{G}_{A}^{down}$\;
\textbf{Return} $\mathbf{x}_A = \{\mathbf{x}_i^*\}_{i\in \mathcal{V}_A}$;
\end{algorithm}

We explain Algorithm \ref{Alg:approximation-2} by a toy example in Fig.~\ref{fig:Paper_example_approximation_algorithm}, where $\mathcal{V}_A = \{2,4\}$. In this case,  $\mathcal{V}_A^{down} = \{2,3,4,5,6\}$, where adversarial nodes can only access the routing information in the dashed frame in Fig.~\ref{fig:Paper_example_approximation_algorithm}(a). 
We then calculate $R[i], \forall i \in \mathcal{V}_A$ based on $\mathcal{G}_A^{down}$, where here we explain the meaning of $R[i]$ using the upstream information to $i$ as in Fig.~\ref{fig:Paper_example_approximation_algorithm}(b) for clarity, while it can be calculated by Algorithm \ref{Alg:calculation_Ri} introduced later without upstream information. We observe that $x_{12}$ of the traffic will arrive to node $2$ without passing the other adversarial node $4$, therefore $R[2] = x_{12}$. For node $4$, the only traffic flow that does not pass node $2$ and arrive to $4$ take the path $1\rightarrow 3\rightarrow 4$, thus $R[4] = x_{13}x_{34}$. 
Finally, we add an arrival rate of $R[i]$ unit to adversarial node $i$ in the downstream subgraph $\mathcal{G}_A^{down}$ as in Fig.~\ref{fig:Paper_example_approximation_algorithm}(c), and we solve the optimal routing attack over $\mathcal{G}_A^{down}$ by Algorithm \ref{Alg:optimal-metric-1-general}. We prove in Theorem \ref{thm:approximation} that $\lambda_{ALG}^*$ under the routing attack $\mathbf{x}_A$ output by Algorithm \ref{Alg:approximation-2} satisfies $\lambda_{ALG}^*\leq 2\lambda_{OPT}^*$.

\begin{figure}[!htbp]
\centering
\includegraphics[width=0.99\linewidth]{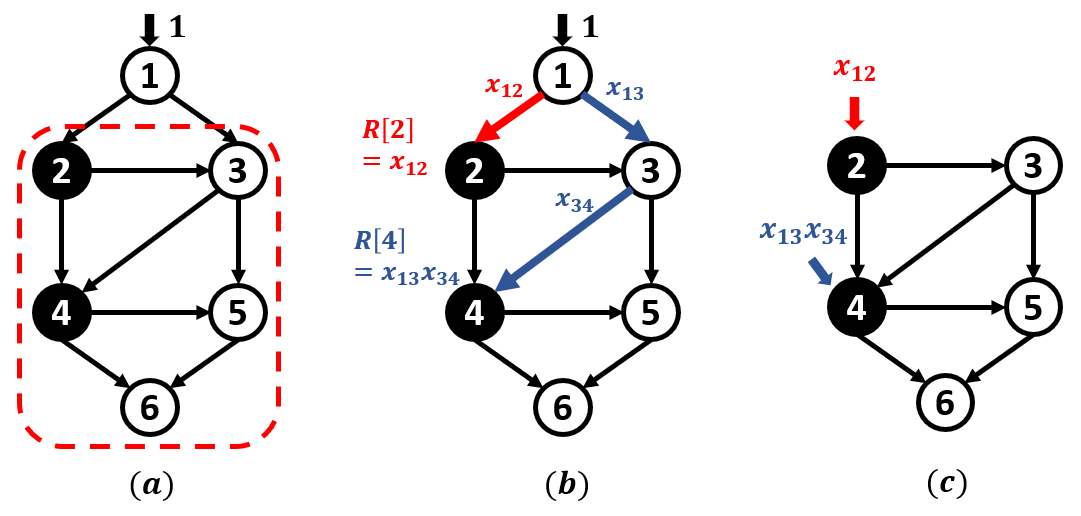}
\caption{Example of running Algorithm \ref{Alg:approximation-2} from (a) to (c)}
\label{fig:Paper_example_approximation_algorithm}
\end{figure}

\begin{theorem}
\label{thm:approximation}
Algorithm \ref{Alg:approximation-2} is a $2$-approximation algorithm.
\end{theorem}

\begin{proof}
Recall the optimal traffic flow vector to \eqref{eqn:min-max-lambda-no-loss-equi} is denoted by $\mathbf{f}^{*}$. We denote $\mathbf{f}^{*}$ by $\mathbf{f}^{OPT}$ in the proof, and the traffic flow vector under the routing output by Algorithm \ref{Alg:approximation-2} by $\mathbf{f}^{ALG}$. Therefore $(\lambda_{OPT}^*)^{-1} = \max_{(i,j)\in \mathcal{E}} f_{ij}^{OPT} / c_{ij}$ and $(\lambda_{ALG}^*)^{-1} = \max_{(i,j)\in \mathcal{E}} f_{ij}^{ALG} / c_{ij}$. Further notice that the flow $f_{ij}$ is the sum of flow that are in the constructed $\mathcal{G}_{A}^{down}$ and those are not. Denote them as $g_{ij}$ and $h_{ij}$ respectively, and thus $f_{ij} = g_{ij} + h_{ij}$. Therefore the approximation ratio
$$
\begin{aligned}
    &\left(\frac{\lambda_{ALG}^*}{\lambda_{OPT}^*}\right)^{-1}  
    = \frac{\max_{(i,j)\in \mathcal{E}} (g_{ij}^{ALG} + h_{ij}^{ALG}) / c_{ij}}{\max_{(i,j)\in \mathcal{E}} (g_{ij}^{OPT} + h_{ij}^{OPT}) / c_{ij}}
    \\& \qquad \overset{(a)}{=} \frac{\max_{(i,j)\in \mathcal{E}} (g_{ij}^{ALG} + h_{ij}^{ALG}) / c_{ij}}{\max_{(i,j)\in \mathcal{E}} (g_{ij}^{OPT} + h_{ij}^{ALG}) / c_{ij}}
    \\&\qquad\overset{(b)}{\geq}  \frac{  \max\left\{\max_{(i,j)\in \mathcal{E}} g_{ij}^{ALG} / c_{ij}, \max_{(i,j)\in \mathcal{E}} h_{ij}^{ALG} / c_{ij}\right\}}{\max_{(i,j)\in \mathcal{E}} (g_{ij}^{OPT} + h_{ij}^{ALG}) / c_{ij}}
    \\&\qquad\overset{(c)}{\geq} \frac{  \max\left\{\max_{(i,j)\in \mathcal{E}} g_{ij}^{OPT} / c_{ij}, \max_{(i,j)\in \mathcal{E}} h_{ij}^{ALG} / c_{ij}\right\}}{\max_{(i,j)\in \mathcal{E}} g_{ij}^{OPT} / c_{ij} + \max_{(i,j)\in \mathcal{E}} h_{ij}^{ALG} / c_{ij}}
    \\&\qquad\overset{(d)}{\geq} 1/2.
\end{aligned}
$$
In the derivation, $(a)$ holds because $h_{ij}$ represents the flow that does not pass any adversarial node any under default routing, thus not related to the difference between adversarial routing solutions and $h_{ij}^{ALG}=h_{ij}^{OPT}$; (b) holds based on $\max_i \{a_i + b_i\} \geq \max\{\max_i\{a_i\}, \max_i\{b_i\}\}$; (c) holds since Algorithm \ref{Alg:approximation-2} outputs the solution that minimizes the no-loss throughput of the constructed $\mathcal{G}_A^{down}$, thus $\max_{(i,j)\in \mathcal{E}} g_{ij}^{ALG} / c_{ij} \geq \max_{(i,j)\in \mathcal{E}} g_{ij}^{OPT} / c_{ij}$, and the denominator is due to $\max_i \{a_i + b_i\} \leq \max_i\{a_i\} + \max_i\{b_i\}$; (d) holds since $\max\{a,b\}/(a+b)\geq 1/2$. Therefore $\lambda_{ALG}^* / \lambda_{OPT}^* \leq 2$.
\end{proof}


We can then obtain Corollary \ref{coro:node_cut}, which states that as long as $\mathcal{V}_A$ contains a \emph{node cut} of $\mathcal{G}$, the removal of which disconnects the source and destination, then Algorithm \ref{Alg:approximation-2} returns the optimal routing attack that minimizes $\lambda^*$ solely with downstream information. 

\begin{corollary}
\label{coro:node_cut}
With $\mathcal{G}_A^{down}$, Algorithm \ref{Alg:approximation-2} outputs the optimal solution to \eqref{eqn:min-max-lambda-no-loss-equi} if $\mathcal{V}_A$ contains a node cut to $\mathcal{G}$.
\end{corollary}

\begin{proof}
If $\mathcal{V}_A$ contains a node cut of the network, then no traffic flows will be sent to destination node without passing any adversarial node. Therefore $h_{ij}^{ALG}=h_{ij}^{OPT}=0, \forall (i,j)\in \mathcal{E}_A^{down}$, and thus  $\lambda_{ALG}^*/\lambda_{OPT}^*=1$.
\end{proof}

We finally introduce Algorithm \ref{Alg:calculation_Ri}, which calculates $R[i]$ for $\forall i \in \mathcal{V}_A$ purely based on $\mathcal{G}_A^{down}$. In Fig.~\ref{fig:Paper_Find_Ri}, we give an example of applying Algorithm \ref{Alg:calculation_Ri} over the network example from Fig.~\ref{fig:Paper_example_approximation_algorithm}. The adversary first does not conduct routing attack on $\mathcal{V}_A=\{2,4\}$, which means nodes in $\mathcal{V}_A$ take their default routing policies before being hijacked. It chooses a timestamp where the network is not overloaded\footnote{We assume that the network before attack does not contain saturated links, which is common in real networks with sufficient provisioned capacity \cite{zhang2021gemini}.} to measure the traffic flow $f_{ij}$ for $\forall (i,j) \in \mathcal{E}_A^{down}$ under the default routing policies of both $\mathcal{V}_N$ and $\mathcal{V}_A$, and the total throughput $F_{total}$ at the destination node. Then we can calculate ${R}[i]$ by traversing nodes in $\mathcal{V}_{A}^{down}$ from upstream to downstream over $\mathcal{G}_{A}^{down}$: Iteratively determine $R[i]$ whenever traversing at $i\in \mathcal{V}_A$, and remove the adversarial node $i$ and its connected downstream links with updated $\mathcal{V}_{A}^{down} = \mathcal{V}_{A}^{down} \backslash \{i\}$ and $\mathcal{E}_{A}^{down} = \mathcal{E}_{A}^{down} \backslash \{(i,j)\in \mathcal{E}, \forall j\}$, and then determine the next adversarial node. 

\begin{algorithm}[!htp]
\caption{Calculate $R[i], i\in \mathcal{V}_A$ based on $\mathcal{G}_A^{down}$}
\label{Alg:calculation_Ri}
\textbf{Input:} $\mathcal{V}_A$, $\mathcal{G}_{A}^{down}$, normal routing $\{\mathbf{x}_{i}\}_{i\in \mathcal{V}_N \cap \mathcal{V}_A^{down}}$\;
Measure the traffic flow $f_{ij}, \forall (i,j)\in \mathcal{E}_A^{down}$ without routing attack on $\mathcal{V}_A$, and total flow $F_{total}$ at the destination node\;
Measure the total traffic arrival to each source node $i\in \mathcal{V}_A^{down}$ as $F[i]$\;
Initialize the total flow to each node $j\in \mathcal{V}_A^{down}$ from within $\mathcal{G}_A^{down}$ as $F[j] = \sum_{i: (i,j)\in \mathcal{E}_A^{down}} f_{ij}$\;
Topological sort $\mathcal{V}_A^{down}$\;
\For{$i \in \mathcal{V}_A^{down}$}{
\textbf{If} $i \in \mathcal{V}_N$, \textbf{continue}\;
$R[i] = F[i] / F_{total}$\;
Remove $i$ and $\{(i,j)\mid (i,j)\in \mathcal{E}_A^{down}\}$ from $\mathcal{V}_A^{down}$ and $\mathcal{E}_A^{down}$ respectively\;
Update $F[j]$ based on $(\mathcal{V}_A^{down}, \mathcal{E}_A^{down})$ by removing flow starting from $i$ to $j$, $\forall j \in \mathcal{V}_A^{down}$\;
}
\textbf{Return} $\{R[i]\}_{i\in \mathcal{V}_A}$;
\end{algorithm}

\begin{figure}[!htbp]
\centering
\includegraphics[width=0.99\linewidth]{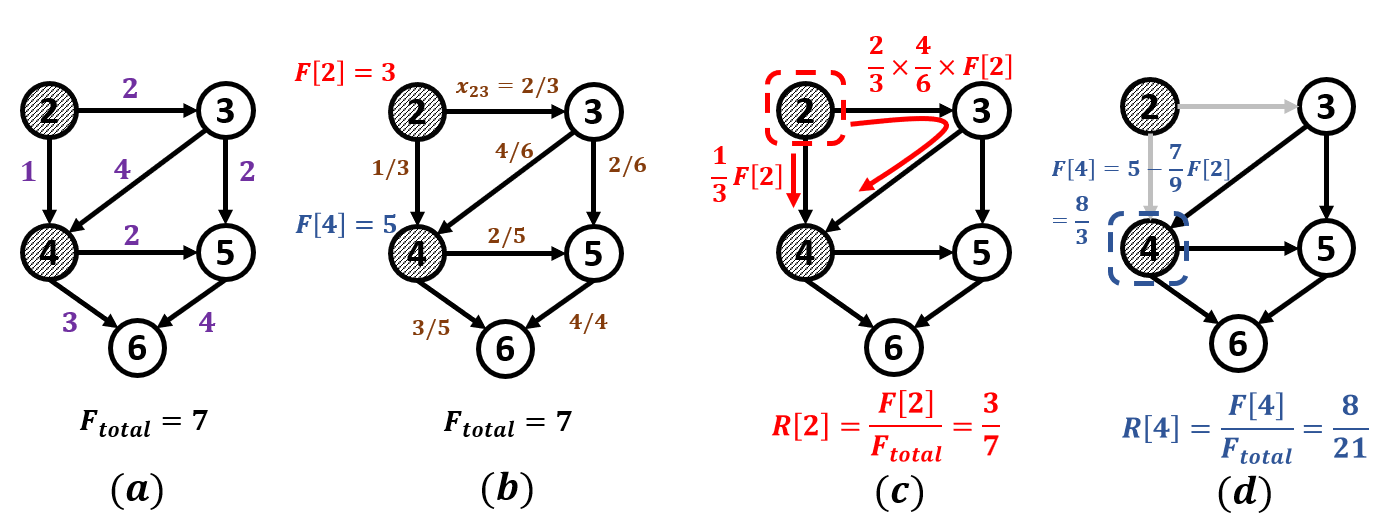}
\caption{\small Example of Algorithm \ref{Alg:calculation_Ri}. We visualize $\mathcal{V}_A$ in Fig.~\ref{fig:Paper_Find_Ri} by shaded instead black nodes since the adversary  does not conduct routing attack when calculating $R[i]$. The original traffic flows are marked by the numbers over each link in (a), and we can measure the total throughput $F_{total}=7$ at the destination node 6. We initialize $F[2]$ and $F[4]$ to be the total flow to node 2 and node 4, which are $3$ and $5$ respectively shown in (b). Since node 2 is the source node in $\mathcal{G}_A^{down}$, $F[2]=3$ does not need to be updated since all the flows from node 2 to 3 and 4 only traverse node 2. Then in (c) we calculate the flows at node 4 that have traversed node 2, which contains paths $2\rightarrow 4$ and $2\rightarrow 3 \rightarrow 4$ with a total of $7/9 \times F[2]$. Then in (d) we substract these flow from $F[4]$ and we get the final $F[4] = 8/21$.}
\label{fig:Paper_Find_Ri}
\end{figure}

\subsection{Practical Extensions}
\label{sec:no-loss-throughput-extension}

\subsubsection{Heuristic for Distributed Attack} Algorithm \ref{Alg:optimal-metric-1-general} and \ref{Alg:approximation-2} require a centralized adversary who can manipulate the routing at different adversarial nodes simultaneously, while it does not apply to distributed routing attack where each adversarial node determines the routing separately. We introduce a heuristic in Algorithm \ref{Alg:heuristic} that works for distributed attack {in directed acyclic networks} with lower time complexity $O(d_{max}|\mathcal{V}_A|M)$. Each $i\in \mathcal{V}_A$ decides its routing that minimizes the no-loss throughput $\lambda^*$ of the induced subgraph formed by  $\{i\} \cup \mathcal{V}_{i}^{down}$ where node $i$ itself serves as the source. We can traverse the nodes from the destination, and whenever we encounter an adversarial node, we decide its routing as above, fix it and continue the traversal to the source. The routing of the predecessor adversarial nodes depends on the successor ones. 

\begin{algorithm}[!htp]
\caption{Heuristics for Distributed Attack}
\label{Alg:heuristic}
\textbf{Input:} directed acyclic $\mathcal{G}=(\mathcal{V}, \mathcal{E})$, $\mathcal{V}_A$, $\{\mathbf{x}_i\}_{i\in \mathcal{V}_N}$\;
Topology sort $\mathcal{V}_A$ from sink to source\;
\For{$i\in \mathcal{V}_A$ in topological sorted order} {
Determine $\mathbf{x}_i^*$ by running Algorithm \ref{Alg:optimal-metric-1-general} on the induced graph of $\mathcal{V}_{i}^{down}$, and fix $\mathbf{x}_i^*$\;
}
\textbf{Return} $\mathbf{x}_A^*$\;
\end{algorithm}

We point out that Algorithm \ref{Alg:heuristic} can lead to an approximation ratio of $2^{|\mathcal{V}_A|}$ under the topology in Fig.~\ref{fig:Paper_example_heuristics_downstream_only}. However, we show empirically in Section \ref{subsec:no_loss_throughput_simulation} that it is an effective heuristic that balances the attack performance and time efficiency.

\begin{figure}[!htbp]
\centering
\includegraphics[width=0.9\linewidth]{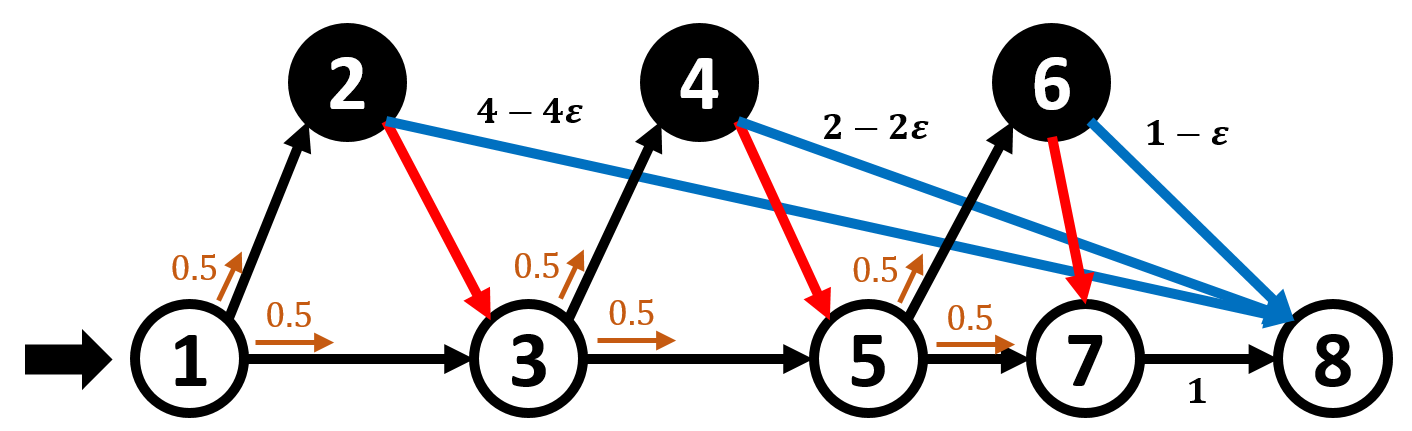}
\caption{\small Example of $2^{|\mathcal{V}_A|}\times \lambda_{OPT}^*$ under Algorithm \ref{Alg:heuristic}. Suppose $\mathcal{V}_{A}=\{2,4,6\}$; $c_{28}=4-4\epsilon$, $c_{48}=2-2\epsilon$, $c_{68}=1-\epsilon$ where $0\leq \epsilon \leq 1$, $c_{78}=1$, and other links have infinite capacity. Nodes 1, 3, and 5 route half of the traffic through each connected downstream link. Algorithm \ref{Alg:optimal-metric-1-general} routes all traffic through red links which leads to $\lambda^* = \lambda_{OPT}^*=1$, while Algorithm \ref{Alg:heuristic} routes all traffic through blue links (topological order $6\rightarrow 4\rightarrow 2$) which leads to $\lambda^* = 8\lambda_{OPT}^* = 8$ when $\epsilon \rightarrow 0$.}
\label{fig:Paper_example_heuristics_downstream_only}
\end{figure}


\subsubsection{Constraints on Routing Attacks}
There may exist constraints over routing attacks in practice for the adversary. For example, the dispatch ratio may need to be within a normal range, i.e. $x_{ij} \in [x_{ij}^{\min}, x_{ij}^{\max}]$ for some $(i,j)\in \mathcal{E}$. Routing in an extreme manner outside such range, like routing all traffic to a single link based on Algorithm \ref{Alg:optimal-metric-1-general}, may have a high risk of being detected as anomalous traffic and thus being exposed \cite{da2016atlantic}. By incorporating such upper and lower bounds on dispatch ratio, the brute-force approach based on Theorem \ref{thm:boundary_property} is no longer guaranteed to be solved in polynomial time even under $|\mathcal{V}_A|=O(1)$ in that the polytope with constraints over each dimension may contain exponential number of vertices{\footnote{Suppose $\mathbf{A}\in \mathbb{R}^{m\times n}$ $(m<n)$ in the constraints $\mathbf{Ax}=\mathbf{b}$ of an LP in standard form, with linearly independent rows, and columns in any subset of $m$ columns of $\mathbf{A}$, then there are at most $\binom{n}{m}$ basic feasible solutions (vertices), each of which corresponds to a basis of the LP. In our problem, $n = O(M)$. Without additional routing constraints, there are only $m = |\mathcal{V}_A| = O(1)$ flow conservation constraints, hence the brute-force algorithm has polynomial-time complexity. However, we may have $m = O(n)$ that leads to $\binom{n}{m} = O(2^M)$ vertices by adding routing constraints.}}, while Algorithm \ref{Alg:optimal-metric-1-general} can still output the optimal solution in polynomial time since $x_{ij} = {f_{ij}}/{\sum_{k:(k,i)\in \mathcal{E}}f_{ki}} \in [x_{ij}^{\min}, x_{ij}^{\max}]$ is a linear constraint given $x_{ij}^{\min}$ and $x_{ij}^{\max}$. 

\section{Loss Maximization}
\label{sec:max_loss}

In this section, we investigate the optimal routing attack $\mathbf{x}_A = \{\mathbf{x}_i\}_{i\in \mathcal{V}_A}$ to maximize the total traffic loss given the traffic arrival rate $\lambda$. We establish that the problem is NP-hard, and
propose two approximation algorithms with performance guarantee for single-hop networks.

\subsection{Problem Formulation and NP-Hardness}

We formulate the loss maximization problem in a general network as in \eqref{eqn:max-loss}.

\begin{equation}
\label{eqn:max-loss}
\begin{aligned}
\min\limits_{\mathbf{f}, \mathbf{x}_A} \quad & \sum_{i: (i,d)\in \mathcal{E}} f_{id} \\
\textrm{s.t.} \quad
& f_{ij} = \min\left\{\Bigg(\sum_{k:(k,i)\in \mathcal{E}}f_{ki}\Bigg)x_{ij},c_{ij}\right\}, \forall (i,j)\in \mathcal{E} 
\\&
\lambda = f_{01}, \qquad f_{ij} \geq 0, ~\forall (i,j)\in \mathcal{E}
\end{aligned}
\end{equation}

The objective of \eqref{eqn:max-loss} is to minimize the total traffic arrived at the destination node, which is equivalent to maximizing the loss $\lambda - \sum_{i: (i,d)\in \mathcal{E}} f_{id}$. The decision variables contain both the flow variables $\mathbf{f}$ and routing policies $\mathbf{x}_A$. The constraints require that the flow over a link $(i,j)$ is equal to the minimum between the total flow injected from node $i$ to $j$ and the link capacity $c_{ij}$. If the flow routed to $(i,j)$ from $i$ exceeds $c_{ij}$, then $f_{ij} = c_{ij}$ and the excess $f_{ij}-c_{ij}$ traffic will be dropped. We point out the loss maximization problem cannot be formulated with linear constraints as \eqref{eqn:min-max-lambda-no-loss}. 
A more compact but less intuitive equivalent way is to remove $\mathbf{x}_A$ from decision variables by expressing $\mathbf{x}_i$ via $\mathbf{f}$ as done in \eqref{eqn:min-max-lambda-no-loss}. 

We prove Proposition \ref{prop:boundary_property_maxloss} which states that the problem \eqref{eqn:max-loss} shares the property introduced in Theorem \ref{thm:boundary_property} for \eqref{eqn:min-max-lambda-no-loss-equi}, indicating that a brute-force approach that follows the idea of Theorem \ref{thm:boundary_property}, with the objective replaced by loss maximization, can output the optimal solution to \eqref{eqn:max-loss}, with time complexity $O(N^{|V_A|}M)$ which is polynomial under $|\mathcal{V}_A|=O(1)$.

\begin{proposition}
\label{prop:boundary_property_maxloss}
There exists an optimal solution to \eqref{eqn:max-loss} such that for $\forall i \in \mathcal{V}_A$, $\exists j$ that $(i,j) \in \mathcal{E}$ and $x_{ij}^*=1$, while $x_{ik}^*=0$ for $k\neq j$.
\end{proposition}

\begin{proof}
Note that the constraint in \eqref{eqn:max-loss} implies that the flow $f_{ij}$ can be expressed by a $k$-layer nested min function by the flows over links $k$ hops predecessor to link $(i,j)$. Consider the network in Fig.~\ref{fig:Paper_example_multihop_networks_and_routing_attack} for an example, the objective function can be written as
$$
\begin{aligned}
f_{46} + f_{56} &= \min\{f_{24} + f_{34}, c_{46}\} + \min\{f_{35}, c_{56}\}
\\&= \min\{\min\{f_{12}, c_{24}\} + \min\{f_{13}x_{34}, c_{34}\}\} 
\\& \quad + \min\{\min\{f_{13}x_{35}, c_{35}\}, c_{56}\} = ...
\end{aligned}
$$

We first consider the case with a single adversarial node $\mathcal{V}_A = \{i\}$. Then the objective in \eqref{eqn:max-loss} is a nested combination of minimum and sum of affine functions of $\mathbf{x}_i$, which is concave since the elementwise $\min$ function is concave. Furthermore, by expressing \eqref{eqn:max-loss} as a function of the decision variable $\mathbf{x}_i$, the constraints are simply $\sum_{j:(i,j)\in \mathcal{E}} x_{ij}=1, x_{ij}\geq 0$, which are linear. Therefore \eqref{eqn:max-loss} is to minimize a concave function in an affine feasible region, indicating that one of the vertices must be optimal, i.e. $\exists j$ that $(i,j) \in \mathcal{E}$ and $x_{ij}^*=1$, while $x_{ik}^*=0$ for $k\neq j$.

Now we extend to general $\mathcal{V}_A$. Note that for $\forall i \in \mathcal{V}_A$, when the routing policies of nodes in $\mathcal{V}_A\backslash \{i\}$ are fixed, the objective is a concave function of the routing policies $\mathbf{x}_i$. Given any fixed routing of nodes in $\mathcal{V}_A\backslash \{i\}$, one of the optimal routing polices at node $i$ must be $x_{ij}=1$ over some link $(i,j)$. This means that given an adversarial routing solution $\{\mathbf{x}_i\}_{i\in \mathcal{V}_A}$, if there exists any node $i_0$ that satisfies $x_{i_0j}<1, \forall (i,j)\in \mathcal{E}$, there must exist $j_0$ where letting $x_{i_0j_0}=1$ will not cause less loss. Therefore through proof by contradiction, there must exist a vertex that maximizes the loss.
\end{proof}

However, we show that unlike no-loss throughput minimization, loss maximization \eqref{eqn:max-loss} is an NP-hard problem under general size of $\mathcal{V}_A$. 

\begin{theorem}
\label{thm:np_hard}
Problem \eqref{eqn:max-loss} is NP-hard under  $|\mathcal{V}_A| = O(N)$.
\end{theorem}

\begin{proof}
We reduce Set Cover to Problem \ref{eqn:max-loss}. Given an instance of Set Cover: $m$ elements $\{e_i\}_{i=1}^m$, $n$ sets $\{S_j\}_{j=1}^n$ where each set covers several elements, what is the minimum number of sets that can fully covers all the elements? We construct the graph in Fig.~\ref{fig:Paper_NP_hard_maxloss}, where element $e_i$ as a node $s_i$, and set $S_j$ as a node $d_j$, and there is a directed link from $(s_i, d_j)$ if and only if the set $S_j$ contains $e_i$, with capacity $c_{s_i,d_j} = \infty$. We further introduce a meta source node $s_0$, and a link from $s_0$ to each of the nodes in $\{s_i\}_{i=1}^m$ with link capacity $c_{s_0, s_i} = \infty$. The routing policy at $s_0$ is $x_{s_0,s_i} = 1/m$, $\forall i = 1,\cdots,m$, and an external arrival rate of $m$. We then introduce a meta destination node $d_0$, and introduce a directed link $(d_j,d_0)$ with capacity $c_{d_j,d_0}=1$. The adversarial node set $\mathcal{V}_A = \{s_i\}_{i=1}^m$, and they can dispatch traffic to any of its connected downstream nodes. Clearly, the optimal routing attack of $\mathcal{V}_A$ that maximizes the loss is to route traffic to minimum number of nodes in $\{d_j\}_{j=1}^n$, where the maximum loss is $m-p$ suppose that the minimum number is $p$. Therefore if there exists a polynomial-time algorithm for Problem \eqref{eqn:max-loss}, then it outputs the minimum number of nodes in $\{d_j\}_{j=1}^n$ that receives traffic under the optimal attack, which corresponds to the minimum set cover over elements $\{e_i\}_{i=1}^m$, thus contradicting to the NP-hardness of Set Cover.
\end{proof}

\begin{figure}[!htbp]
\centering
\includegraphics[width=0.85\linewidth]{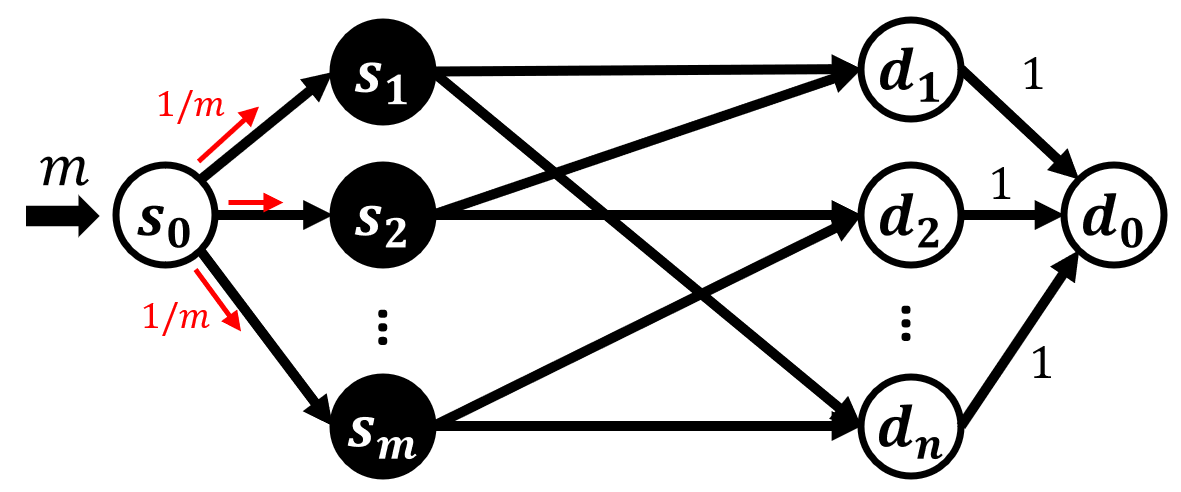}
\caption{Graph Construction from Set Cover}
\label{fig:Paper_NP_hard_maxloss}
\end{figure}

\subsection{Approximation Algorithms}
\label{subsec:apprimaxtion_algorithm}

We propose two approximation algorithms with multiplicative and additive performance guarantees respectively in single-hop networks, where routing attacks are conducted on ingress nodes $\mathcal{V}_S$. Under static routing, we can assume w.l.o.g. $\mathcal{V}_S = \mathcal{V}_A$, i.e., all ingress nodes are hijacked. This is because if an ingress node is normal, then based on its default routing policy we can calculate the traffic rate sent from this node to each of its connected downstream egress nodes. We can then correspondingly reduce the service rate $\mu_j$ at each egress node $d_j$ by the total amount of traffic rates from all the normal ingress nodes to $d_j$. Therefore, the single-hop network with normal ingress nodes can be equivalently transformed into the one with all ingress nodes being adversarial, and updated remaining service rates at $\mathcal{V}_D$. We explain the algorithm design considering $\mathcal{V}_S = \mathcal{V}_A$ below.

\subsubsection{Algorithm with Multiplicative Guarantee} We propose a greedy-based algorithm that has a worst-case approximation ratio of $1/\sqrt{|\mathcal{V}_A|}$. We summarize the details in Algorithm \ref{Alg:max_loss_approx_multiplicative}, which executes two greedy approaches and outputs the solution that leads to higher loss. Both approaches share the idea of iteratively routing traffic to the egress node that leads to maximum overload.
Approach 1 maximizes the overload at an egress node \emph{without} normalization. In each iteration, the adversary takes aim at the egress node $d_j^* = \argmax_{d_j \in \mathcal{V}_D} 
\left(\sum_{s_i \in UD: (s_i, d_j)\in \mathcal{E}} \lambda_i\right) - \mu_j$, where $UD$ denotes the set of adversarial nodes whose routing is undecided by the current iteration, and then sets the routing policy at all ingress nodes in $UD \cap \{s_i\in \mathcal{V}_S \mid (s_i, d_j^*) \in \mathcal{E}\}$ to be dispatching all their traffic to $d_j^*$. 
Approach 2 maximizes overload at an egress node \emph{with} normalization, where at each $d_j \in \mathcal{V}_D$, the adversary finds the routing that maximizes the per-ingress overload at $d_j$, defined as
\begin{equation}
\label{eqn:PSO}
    PSO[d_j] := \max_{\mathcal{S}_{d_j} \subseteq  UD} \frac{1}{|\mathcal{S}_{d_j}|}
\left(\sum_{s_i\in \mathcal{S}_{d_j}: (s_i, d_j)\in \mathcal{E}} \lambda_i - \mu_j\right)
\end{equation}
which means to find the subset of ingress nodes, denoted by $\mathcal{S}_{d_j}$, that can send traffic to and maximizes the overload at $d_j$ normalized by the size of $\mathcal{S}_{d_j}$. Denote the optimal $\mathcal{S}_{d_j}$ with respect to \eqref{eqn:PSO} for node $d_j$ by $\mathcal{S}_{d_j}^{*}$. The adversary routes all the traffic at ingress nodes in $\mathcal{S}_{d_j^*}^{*}$ to $d_j^* = \argmax_{d_j\in \mathcal{V}_D} PSO[d_j]$. We iterate the above process over the ingress nodes whose routing policies are yet to be determined. Note that the determination of $\mathcal{S}_{d_j}^*$ for each $d_j\in \mathcal{V}_D$ can be done in polynomial time: Initialize $\mathcal{S}_{d_j}^*=\emptyset$, sort all ingress nodes in $UD \cap \{s_i\in \mathcal{V}_S \mid (s_i, d_j^*) \in \mathcal{E}\}$ non-increasingly with respect to the traffic arrival rates to them, and add these ingress nodes to $\mathcal{S}_{d_j}^*$ in sequel, until the $PSO[d_j]$ reaches the peak value and starts decreasing. 


\begin{algorithm}[!htp]
\caption{Approximation Algorithm with Multiplicative Guarantee}
\label{Alg:max_loss_approx_multiplicative}
\textbf{Input:} Single-hop network $\mathcal{G}=(\mathcal{V}, \mathcal{E})$, $\mathcal{V}_A$, $\boldsymbol{\lambda}$, $\boldsymbol{\mu}$;

\SetKwFunction{Approach}{GreedyAlg}

\SetKwProg{Fn}{Function}{\string:}{end}

\Fn{\Approach{$\mathcal{G}$, $\boldsymbol{\lambda}$, $\boldsymbol{\mu}$, GreedyType}}{
    Initialize $\mathbf{x}_i = \boldsymbol{0},~\forall s_i \in \mathcal{V}_A$; 
    $UD = \mathcal{V}_A$\;
    \While{$\exists s_i\in UD$}{
    \If{$GreedyType$ is Approach1}
    {\small $j^* \leftarrow \arg\max_j \left(\sum\limits_{s_i\in UD, (s_i,d_j)\in \mathcal{E}} \lambda_i\right) - \mu_j$\;
    }
    \If{$GreedyType$ is Approach2}{
     $j^* \leftarrow \arg\max_j PSO[d_j]$ as  \eqref{eqn:PSO}\;
    }
    Let $\mathbf{x}_{ij^*} = 1$, $UD = UD\backslash {s_i}$, $\forall s_i, (s_i, d_j^*)\in \mathcal{E}$\;
    }
    Calculate $loss$ based on $sol \leftarrow \{\mathbf{x}_i\}_{i\in \mathcal{V}_A}$\;
    \Return{$loss$, $sol$}
}
Update remaining service rates at $\mathcal{V}_D$\;
Remove $\mathcal{V}_N$ from $\mathcal{V}_S$, and associated links from $\mathcal{E}$\;
$loss_1$, $sol_1$ $\leftarrow$ \Approach($\mathcal{G}$, $\boldsymbol{\lambda}$, $\boldsymbol{\mu}$, Approach1)\;
$loss_2$, $sol_2$ $\leftarrow$ \Approach($\mathcal{G}$, $\boldsymbol{\lambda}$, $\boldsymbol{\mu}$, Approach2)\;
\textbf{Return} $sol_1$ if $loss_1 > loss_2$ else $sol_2$\;
\end{algorithm}

We demonstrate in Theorem \ref{thm:multiplicative} that taking the better routing attack solution between Approach 1 and 2 leads to a worst-case approximation ratio of $1/\sqrt{|\mathcal{V}_A|}$. The proof is deferred to the appendix. A single application of either of the two approaches, however, has a worst-case approximation ratio of $1/|\mathcal{V}_A|$. We give toy examples in Fig.~\ref{fig:Paper_example_approximation_multiplicative} to validate the tightness of these performance guarantees. The red routing is the optimal routing attack, while the blue one is the routing attack output by Approach 1 in Fig.~\ref{fig:Paper_example_approximation_multiplicative}(a), Approach 2 in Fig.~\ref{fig:Paper_example_approximation_multiplicative}(b), and Algorithm \ref{Alg:max_loss_approx_multiplicative} in Fig.~\ref{fig:Paper_example_approximation_multiplicative}(c), respectively. In the table in Fig.~\ref{fig:Paper_example_approximation_multiplicative}, $\Delta_{OPT}$, $\Delta_{A1}$, $\Delta_{A2}$ and $\Delta_{MUL}$ denote respectively the loss under optimal routing attack, applying Approach 1, applying Approach 2, and applying Algorithm \ref{Alg:max_loss_approx_multiplicative} which has multiplicative guarantee. 

\begin{theorem}
\label{thm:multiplicative}
Algorithm \ref{Alg:max_loss_approx_multiplicative} outputs a solution with loss $\Delta_{MUL}$ that satisfies
$\frac{\Delta_{MUL}}{\Delta_{OPT}} \geq \frac{1}{\sqrt{|\mathcal{V}_A|}}$. 
\end{theorem}
\begin{figure}[!htbp]
\centering
\includegraphics[width=0.92\linewidth]{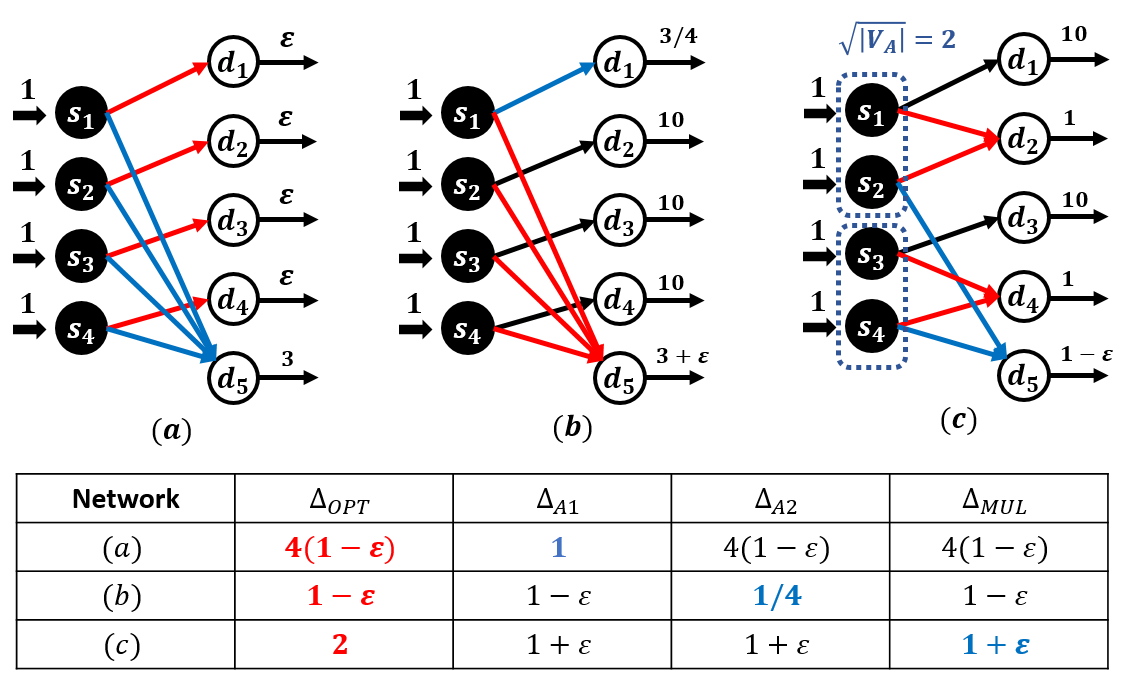}
\caption{\small Tightness validation of Theorem \ref{thm:multiplicative} under $|\mathcal{V}_A|=4$. Consider $\epsilon \rightarrow 0$, Approach 1 on network (a) and Approach 2 on network (b) leads to $1/|\mathcal{V}_A|=1/4$ approximation ratio; Algorithm \ref{Alg:max_loss_approx_multiplicative} gives worst-case approximation ratio $1/\sqrt{|\mathcal{V}_A|}=1/2$ on network (c).}
\label{fig:Paper_example_approximation_multiplicative}
\end{figure}


\subsubsection{Algorithm with Additive Guarantee} We further propose Algorithm \ref{Alg:max_loss_approx_additive} with additive performance guarantee.
The main idea is to iteratively solve the $\lambda^*$-minimization problem \eqref{eqn:min-max-lambda-no-loss-equi} until either all the adversarial nodes have determined their routing policies or the $\lambda^*$ exceeds the given arrival rate $\lambda$. The intuition behind Algorithm \ref{Alg:max_loss_approx_additive} is to greedily find the routing attack to cause the next possible saturated link with minimum arrival rate increment, which can be done by iteratively fixing the routing decided in previous iterations (line 7), and updating the capacity of previous saturated links to be infinite (line 8).

\begin{algorithm}[!htp]
\caption{Approximation Algorithm with Additive Guarantee}
\label{Alg:max_loss_approx_additive}
\textbf{Input:} Single-hop network $\mathcal{G}=(\mathcal{V}, \mathcal{E})$, $\mathcal{V}_A$, $\boldsymbol{\lambda}$, $\boldsymbol{\mu}$\;
Introduce meta-source normal node $s_0$ to $\mathcal{G}$ with dispatch ratio $x_{s_0, s_i} = \lambda_i / \sum_{k=1}^{|\mathcal{V}_S|}\lambda_k, \forall s_i$\;
Introduce meta-destination $d_0$ with $c_{d_j,d_0} = \mu_j$, $\forall d_j$\;
Initialize $\mathbf{x}_i = \boldsymbol{0}, \forall s_i \in \mathcal{V}_A$; $UD = \mathcal{V}_A$\;
\While{$\exists s_i\in UD$} {
Run Algorithm \ref{Alg:optimal-metric-1-general} on the updated $\mathcal{G}, \mathcal{V}_A, \mathbf{x}_N$ and get the first saturated link $(d_{j^*}, d_0)$\;
Let $x_{ij^*} = 1$,  $UD = UD \backslash \{s_i\}$, and fix the routing at $s_i$, for $\forall s_i \in UD, (s_i, d_j^*)\in \mathcal{E}$\;
Let $c_{d_{j^*}, d_0} = \infty$\;
}
Calculate $loss$ based on $sol \leftarrow \{\mathbf{x}_i\}_{s_i\in \mathcal{V}_A}$\;
\textbf{Return} $sol$, $loss$\;
\end{algorithm}

We demonstrate in Theorem \ref{thm:additive} that the gap between the output of Algorithm \ref{Alg:max_loss_approx_additive} and the maximum loss given by \eqref{eqn:max-loss} is bounded above by $\lambda / 4$ in the worst case for any single-hop network. The proof is deferred to the appendix. This means, for example, if the optimal routing attack causes $50\% \times \lambda$ traffic loss, then the output attack of Algorithm \ref{Alg:max_loss_approx_additive} at least causes $25\% \times \lambda$ traffic loss.

\begin{theorem}
\label{thm:additive}
Algorithm \ref{Alg:max_loss_approx_additive} outputs a solution with loss $\Delta_{ADD}$ that satisfies 
$\frac{\Delta_{OPT} - \Delta_{ADD}}{\lambda} \leq \frac{1}{4}$.
\end{theorem}

We point out that the $1/4$ bound is tight, given by the example shown in Fig.~\ref{fig:Paper_example_approximation_additive}. The optimal routing at $s_2$ is to route traffic to $d_1$, causing loss of $2$, while the solution by Algorithm \ref{Alg:max_loss_approx_additive} will route traffic to $d_2$, causing loss of $1 + \epsilon$ where $\epsilon > 0$. Thus the gap $(2-(1 + \epsilon))/\lambda \rightarrow 1/\lambda = 1/4$ when $\epsilon \rightarrow 0$. However, the condition to achieve the worst-case gap $1/4$ is very strong, which requires the arrival rates to all the ingress nodes and the service rates of all the egress nodes to meet some ratio equality, where in Fig.~\ref{fig:Paper_example_approximation_additive} they are $x_{s_0,s_2} = \mu_{2}/\mu_{1}$ and $\mu_{2}/\mu_{1} \rightarrow 1/2$. 
We show in Section \ref{sec:evaluation} that empirically Algorithm \ref{Alg:max_loss_approx_additive} leads to the gap less than $5\% \times \lambda$ in more than $95\%$ tested cases. 

\begin{figure}[!htbp]
\centering
\includegraphics[width=0.85\linewidth]{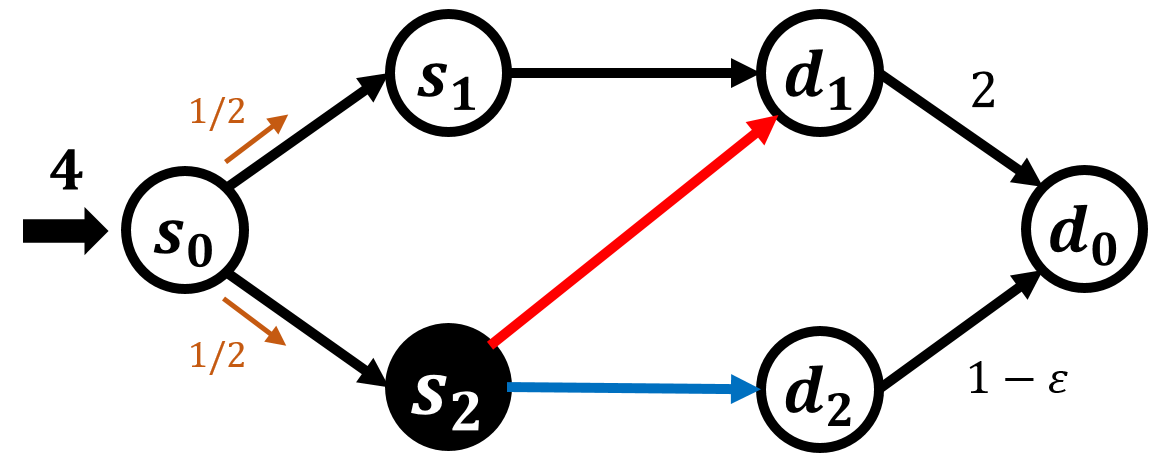}
\caption{Tightness validation of Theorem \ref{thm:additive} ($\lambda = 4$, $\mathcal{V}_A = \{s_2\}$): red route is optimal, blue route is the output of Algorithm \ref{Alg:max_loss_approx_additive}.}
\label{fig:Paper_example_approximation_additive}
\end{figure}


\subsection{Practical Extensions} 


\subsubsection{Constraints on Routing Attack} We follow the constraint form, $x_{ij} \in [x_{ij}^{\min}, x_{ij}^{\max}]$ for some $(s_i,d_j)\in \mathcal{E}$, in Section \ref{sec:no-loss-throughput-extension}. We show that both Algorithm \ref{Alg:max_loss_approx_multiplicative} and \ref{Alg:max_loss_approx_additive} can incorporate this constraint in single-hop networks. For the lower bound, we can require that each adversarial ingress node $s_i$ initially sends $\lambda_i \times x_{ij}^{\min}$ traffic through each link $(s_i, d_j)$. For the upper bound, Algorithm \ref{Alg:max_loss_approx_additive}, which is based on no-loss throughput minimization, can keep unchanged as discussed in Section \ref{sec:no-loss-throughput-extension}. {In Algorithm \ref{Alg:max_loss_approx_multiplicative}, an adversarial node $s_i$ will still be considered in future iterations if $s_i$ has not determined the egress nodes for all its traffic, unlike the unconstrained case where the routing of an adversarial node will be determined in one of the iterations.
For example, suppose $x_{ij}^{\min}=0$ and consider that $s_i$ routes $x_{s_id_j}^{\max}$ of all the traffic to $d_j$, then there remains $\lambda_i(1 - x_{s_id_j}^{\max})$ units of traffic rates at $s_i$ whose routing is undecided, and we will consider $s_i$ in future iterations to route the remaining traffic to egress nodes other than $d_j$ until there no longer exists traffic with routing undecided at $s_i$. With such adjustment, the factor $1/|\mathcal{S}_{d_j}|$ in \eqref{eqn:PSO} should be the reciprocal of the sum of the remaining proportions of traffic with undecided routing over $\mathcal{S}_{d_j}$.}

\subsubsection{Multi-hop Networks} We briefly discuss the extension of the algorithms to multi-hop networks. In Algorithm \ref{Alg:max_loss_approx_multiplicative}, unlike in single-hop networks we maximize the inflow to each node in a greedy manner, in multi-hop networks we greedily find the node $i$ that achieves maximum possible loss over its connected downstream links, under the adversarial routing at $\mathcal{V}_A \cap \mathcal{V}_i^{up}$ that maximizes the total traffic inflow to node $i$. This is similar to lines 2-5 in Algorithm \ref{Alg:optimal-metric-1-general} with the objective changed to traffic loss. 
For Algorithm \ref{Alg:max_loss_approx_additive}, we can naturally extend it to multi-hop networks based on Algorithm \ref{Alg:optimal-metric-1-general}: iteratively saturate the next link by minimizing the no-loss throughput $\lambda^*$, until $\lambda^* \geq \lambda$ or all nodes in $\mathcal{V}_A$ have determined their routing policies. 



\section{Optimal Node Selection}
\label{sec:optimal_node_selection}

In this section, we investigate the problem where the adversary needs to select a limited number of nodes to hijack before conducting routing attack. The discussion serves as an extension to the capability of routing attack, which informs the most critical nodes that should be protected from routing attack. 

\subsection{Problem Formulation and NP-Hardness}

We consider a set of candidate nodes that can be hijacked by the adversary, $\mathcal{V}_{cand}$. The adversary selects at most $K \leq |\mathcal{V}_{cand}|$ nodes to hijack, and the objective is to conduct routing attack on them to minimize no-loss throughput or maximize loss, i.e., solving \eqref{eqn:min-max-lambda-no-loss-equi} or \eqref{eqn:max-loss}. This setting maps to the practice: $\mathcal{V}_{cand}$ is a subset of nodes without secure settings so that the adversary can hijack, and the adversary can only attack a limited number of nodes due to the cost and the risk to be detected. The problem can be extended naturally to a weighted version where different nodes have different costs to be hijacked.

We show Theorem \ref{thm:np_hard_node_selection} that for both no-loss throughput minimization and loss maximization, the problem is NP-hard for general $K=O(N)$, while solvable in polynomial time for $K=O(1)$. The proof is deferred to the appendix.

\begin{theorem}
\label{thm:np_hard_node_selection}
Finding the optimal selection of $K$ nodes from $\mathcal{V}_{cand}$ to conduct routing attack for no-loss throughput minimization and loss maximization is NP-hard under $K=O(N)$, while there exists a polynomial-time brute-force algorithm under $K=O(1)$.
\end{theorem}

Although polynomial-time under $K=O(1)$, the brute-force solution introduced in the proof of Theorem \ref{thm:np_hard_node_selection}
is impractical when $K$ is a large constant. Below we propose efficient heuristic algorithms which can minimize the no-loss throughput exactly, and achieve close-to-optimal traffic loss verified empirically in Section \ref{sec:simulation_optimal_node_selection}, when $\mathcal{V}_{cand}$ are in \emph{parallel} structure. Parallel $\mathcal{V}_{cand}$ means any node in $\mathcal{V}_{cand}$ is not upstream or downstream to another node in $\mathcal{V}_{cand}$. The single-hop network discussed in Section \ref{subsec:apprimaxtion_algorithm} is a special case where $\mathcal{V}_{cand}\subseteq \mathcal{V}_S$ are parallel to each other. 

\subsection{Algorithms Under Parallel $\mathcal{V}_{cand}$}


\subsubsection{No-Loss Throughput Minimization} We show that the optimal node selection problem is solvable in polynomial time with parallel $\mathcal{V}_{cand}$ for no-loss throughput minimization by Algorithm \ref{Alg:optimal_node_selection_parallel}. Denote the candidate nodes upstream to node $i$ by $\mathcal{V}_{cand}^{up,i}$ (assume $i\in \mathcal{V}_{cand}^{up,i}$ if $i\in \mathcal{V}_{cand}$).
The main idea is to traverse each link $(i,j)$ in the downstream graph to $\mathcal{V}_{cand}$. 
For each candidate node $v\in \mathcal{V}_{cand}^{up,i}$, the adversary can evaluate the maximum reduction of the required traffic arrival rate at the source node to saturate $(i,j)$ by choosing to attack $v$. The adversary then chooses $\min\{K, |\mathcal{V}_{cand}^{up,i}|\}$ nodes from $\mathcal{V}_{cand}^{up,i}$ non-decreasingly in terms of the maximum reduction. Denote the choice for $(i,j)\in \mathcal{E}$ as $\mathcal{V}_{A}^{(i,j)}$. We can traverse all links $(i,j)$ downstream to $\mathcal{V}_{cand}$ to build $\{\mathcal{V}_{A}^{(i,j)}\}_{(i,j)\in \mathcal{E}}$, and pick the $\mathcal{V}_A^{(i^*,j^*)}$ under which the no-loss throughput is minimized by solving the routing attack for \eqref{eqn:min-max-lambda-no-loss-equi}.
We prove that it gives the exact optimal node selection in Theorem \ref{thm:node_selection_parallel}, where the proof is deferred to the appendix.

\begin{algorithm}[!htp]
\caption{Optimal node selection for no-loss throughput minimization under parallel $\mathcal{V}_{cand}$}
\label{Alg:optimal_node_selection_parallel}
\textbf{Input:} $\mathcal{G}=(\mathcal{V}, \mathcal{E})$, 
parallel $\mathcal{V}_{cand}$, number of nodes to attack $K$,  default routing $\mathbf{X}$\; 
Build $\mathcal{V}_{cand}^{up,i}$ for each node $i \in \mathcal{V}_{cand}$ and downstream to $\mathcal{V}_{cand}$\;
\For{$(i,j)\in \mathcal{E}$ downstream to $\mathcal{V}_{cand}$}{
Calculate $\Delta_{v}^{(i,j)}, ~\forall v\in \mathcal{V}_{cand}^{up,i}$, the difference of minimum arrival rate that saturates a $(i,j)$ with and without attacking $v$, not attacking other candidate nodes\;
Sort non-increasingly $\{\Delta_{v}^{(i,j)}\}_{v\in \mathcal{V}_{cand}^{up,i}}$\;
Determine $\mathcal{V}_A^{(i,j)}$, the optimal node choice to saturate link $(i,j)$, by finding the top-$\min\{K,|\mathcal{V}_{cand}^{up,i}|\}$ candidate nodes in $\mathcal{V}_{cand}^{up,i}$ with maximum $\Delta_{v}^{(i,j)}$\;
}
Calculate the minimum no-loss throughput under routing attack over $\mathcal{V}_A^{{(i,j)}}$ for all $(i,j)\in \mathcal{E}$ downstream to $\mathcal{V}_{cand}$, and find the minimum one corresponding to link $(i^*,j^*)$\;
\textbf{Return} $\mathcal{V}_A^{(i^*,j^*)}$ as the selected nodes to attack\;
\end{algorithm}

\begin{theorem}
\label{thm:node_selection_parallel}
Algorithm \ref{Alg:optimal_node_selection_parallel} outputs the optimal node selection to conduct routing attack that leads to minimum no-loss throughput given parallel $\mathcal{V}_{cand}$.
\end{theorem}


\subsubsection{Loss Maximization} Unlike no-loss throughput, it is challenging to derive heuristics with performance guarantee even under parallel $\mathcal{V}_{cand}$ for loss maximization. The main bottleneck is that evaluating the performance of a specific node selection, i.e., finding the optimal routing attack to achieve maximum loss among a given set of the selected nodes, is NP-hard. 
We leave the theoretical challenge in future work, while we propose a heuristic in Algorithm \ref{Alg:optimal_node_selection_max_loss} over single-hop networks with $\mathcal{V}_{cand}\subseteq \mathcal{V}_S$ to explain the basic idea by a greedy approach over each egress node, where the extension to parallel $\mathcal{V}_{cand}$ is natural by a greedy approach over each link in the network. Algorithm \ref{Alg:optimal_node_selection_max_loss} iteratively picks nodes to attack until $K$ nodes have been chosen. In each iteration, for each egress node $d_j$, the algorithm evaluates the \emph{value} for each of its connected ingress node $s_i \in UD$, i.e. whose routing has not been decided, to route all its traffic to $d_j$. The value is defined as $v_{ij} = \lambda_i(1-x_{ij})$, the increment of inflow to $d_j$ after routing attack where $x_{ij}$ is the default dispatch ratio from $s_i$ to $d_j$. Then for each egress node $d_j$, the adversary attacks the nodes that are not chosen in previous iterations based on the non-increasing order of $\{v_{ij}\}_{s_i: (s_i,d_j)\in \mathcal{E}}$ 
until either the per-ingress contribution to the loss at $d_j$, defined as $PSC[d_j]$ in Algorithm \ref{Alg:optimal_node_selection_max_loss}, reaches the peak value denoted by $PSC^*[d_j]$\footnote{The process is similar to the \eqref{eqn:PSO} in Algorithm \ref{Alg:max_loss_approx_multiplicative}.}, or the algorithm has chosen $K$ nodes in total till this iteration. 
With all $PSC^*[d_j]$'s calculated and their corresponding chosen ingress nodes to attack, denoted by $\mathcal{S}_j$, 
the algorithm can then determine the new nodes to attack in this iteration to be $\mathcal{S}_{j^*}$ where $d_{j^*} = \arg\max_{d_j\in \mathcal{V}_D} PSC^*[d_j]$, and determine the routing policies of nodes in $\mathcal{S}_{j^*}$ to be dispatching all the traffic to $d_{j^*}$. The algorithm then fixes the routing of the selected ingress nodes and updates the remaining service rate of the $d_{j^*}$, and iteratively runs the above process. 

\begin{algorithm}[!htp]
\caption{Heuristic node selection for loss maximization in single-hop networks}
\label{Alg:optimal_node_selection_max_loss}
\textbf{Input:} Single-hop network $\mathcal{G}$, $\boldsymbol{\lambda}$, $\boldsymbol{\mu}$,  $\mathcal{V}_{cand}$, number of nodes to attack $K$, default routing $\mathbf{X}$\; 

\While{$|Selected| < K$}{
\For{$d_j \in \mathcal{V}_D$}{
Evaluate \emph{value} for each $s_i\in \mathcal{V}_S \backslash Selected$ as $v_{ij} = \lambda_i(1-x_{ij})$\;
Sort nodes in $\mathcal{V}_S \backslash Selected$ non-increasingly with respect to \emph{value}\;
Initialize $S_{tmp} = \emptyset$\;
Add each ingress node to $d_j$ in the above sorted order to  $S_{tmp}$ and update $PSC[d_j] := \frac{1}{|S_{tmp}|}\left(\sum_{s_i\in S_{tmp}} v_{ij} - \mu_j\right)$ before this value starts to decrease\;
Find the peak value $PSC^*[d_j]$ in line 7, and the selected ingress nodes as $\mathcal{S}_j$\;
}
Determine $j^* = \arg\max_{d_j \in \mathcal{V}_D} PSC^*[d_j]$, and add $\mathcal{S}_{j^*}$ to $Selected$. Updated the routing of nodes in $\mathcal{S}_{j^*}$ in $\mathbf{X}$, and the remaining service rate $\mu_{j^*}$\;
}
\textbf{Return} $Selected$ as the selected nodes, $\mathbf{X}$ as the updated routing of $Selected$\;
\end{algorithm}



\section{Performance Evaluation}
\label{sec:evaluation}


In this section, we evaluate the proposed routing attack and optimal node selection algorithms to showcase their near-optimal performance over no-loss throughput minimization and loss maximization over a wide range of network settings, including different network densities, sizes of adversarial node set, and default routing policies. 


\begin{table*}[!htbp]
    \centering
{\fontsize{10}{12}\selectfont
\begin{tabularx}{\textwidth}
{cc|*{3}{>{\centering\arraybackslash}X}|*{3}{>{\centering\arraybackslash}X}|*{3}{>{\centering\arraybackslash}X}|*{3}{>{\centering\arraybackslash}X}}
    \toprule
    \multicolumn{2}{c|}{} & \multicolumn{3}{c|}{Uniform} & \multicolumn{3}{c|}{Proportional} & \multicolumn{3}{c|}{ECMP} & \multicolumn{3}{c}{MaxFlow} \\
    Den & $|\mathcal{V}_A|$ &        Mean &         $90\%$ &         Max &         Mean &         $90\%$ &         Max &        Mean &         $90\%$ &         Max &        Mean &         $90\%$ &         Max \\
    \midrule
    0.4 & 10 &    1.00 &  1.00 &  1.09 &         1.01 &  1.00 &  1.18 &  1.00 &  1.00 &  1.15 &    1.01 &  1.00 &  2.00 \\
    0.8 & 10 &    1.01 &  1.07 &  1.24 &         1.02 &  1.10 &  1.27 &  1.01 &  1.08 &  1.15 &    1.04 &  1.17 &  1.50 \\
    0.4 & 20 &    1.00 &  1.00 &  1.03 &         1.00 &  1.00 &  1.04 &  1.00 &  1.00 &  1.01 &    1.00 &  1.00 &  1.25 \\
    0.8 & 20 &    1.00 &  1.00 &  1.03 &         1.00 &  1.00 &  1.04 &  1.00 &  1.00 &  1.02 &    1.01 &  1.07 &  1.33 \\
    \bottomrule
    \end{tabularx}}
    \caption{Approximation ratio statistics for Algorithm \ref{Alg:approximation-2} ($90\%$ means 90-percentile)}
    \label{tab:approximation-2}
\end{table*}

\begin{figure*}[!htbp]
\centering
\includegraphics[width=1.0\linewidth]{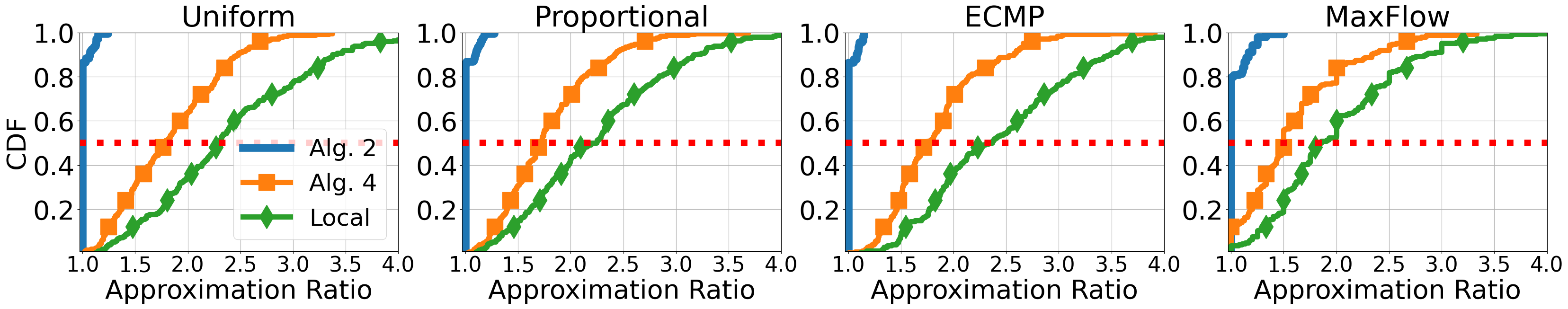}
\caption{CDFs of approximation ratio under density $0.8$ and $|\mathcal{V}_A|=20$}
\label{fig:N50_approximation_ratio_different_default_routing}
\end{figure*}

\subsection{No-Loss Throughput $\lambda^*$ Minimization}
\label{subsec:no_loss_throughput_simulation}

For no-loss throughput minimization, we have shown that Algorithm \ref{Alg:optimal-metric-1-general} yields the exact optimal solution in polynomial time. This exact algorithm serves as our baseline, and we  test the performance of 3 algorithms by measuring their gaps to the baseline: (i) 
the 2-approximation Algorithm \ref{Alg:approximation-2}, in order to validate its closeness to the baseline; (ii) the distributed attack heuristic Algorithm \ref{Alg:heuristic}, in order to evaluate the performance gap between distributed and centralized routing attack; (iii) a \emph{local heuristic} where each adversarial node routes all the traffic through the link with minimum capacity among all its connected downstream links, in order to demonstrate the performance enhancement that adversaries can achieve via Algorithm \ref{Alg:optimal-metric-1-general} compared with a naive attack.

We simulate multi-hop networks with size $|\mathcal{V}|=50$. Our evaluation has the following dimensions. (I) Network Density: We generate random network topologies\footnote{We ensure that there exists at least a path from source to destination.} given different link existence probabilities $p$, and present representative results under $p=0.4$ and $p=0.8$; (II) Number of Adversarial Nodes $|\mathcal{V}_A|$: We showcase results for $|\mathcal{V}_A|=10$ and $20$, where other values exhibit similar performance; (III) Default Routing Policy: We consider 4 routing policies at normal nodes, which have been applied in different scenarios: (i) Uniform routing (random packet spray): Each normal node $i$ dispatches an equal portion of traffic to each of its connected downstream nodes, oblivious to link capacity. (ii) Capacity Proportional routing (weighted packet spray): The dispatch ratio of traffic to a downstream connected node is proportional to the capacity of this link, i.e., $x_{ij}=c_{ij}/\sum_{k: (i,k)\in \mathcal{E}} c_{ik}$. (iii) Equal-Cost Multi-Path (ECMP) routing \cite{zhou2014wcmp}: Traffic at the source is dispatched through the paths with minimum cost. If there are $L>1$ min-cost paths, then $1/L$ of the traffic takes each of the paths. We only present the results where each link has the same cost, as the other costs we have tested share similar results. ECMP is widely used in load balancing and commonly applied in industrial data center and cloud networks \cite{singh2015jupiter}. (iv) MaxFlow routing \cite{altner2010maximum}: This default routing achieves maximum throughput when there is no attack, which is widely studied in theory and serves as a paradigm for network design \cite{tassiulas1990stability}. 
We simulate $10,000$ network instances for each combination of the above network settings: We construct $20$ different topologies, and under each topology, we consider $20$ different adversarial node sets $\mathcal{V}_A$, and further under each $\mathcal{V}_A$, we set $25$ different link capacity allocations. 

We evaluate the cumulative distribution function (CDF) of the approximation ratio $\lambda_{ALG}^* /\lambda_{OPT}^*$, where $\lambda_{ALG}^*$ is the no-loss throughput output by the tested algorithms and $\lambda_{OPT}^*$ is the optimal routing attack by Algorithm \ref{Alg:optimal-metric-1-general}. We present the approximation ratio statistics of the proposed Algorithm \ref{Alg:approximation-2} in Table \ref{tab:approximation-2}, including the mean, 90-percentile, and maximum approximation ratio under various network settings. For instance, the value $1.09$ in the first row signifies the maximum approximation ratio of Algorithm \ref{Alg:approximation-2} under density $0.4$, $|\mathcal{V}_A|=10$, and uniform default routing, over all $10,000$ tested instances under this setting. We have the following observations: (i) In any tested setting, Algorithm \ref{Alg:approximation-2} exhibits near-optimal performance in most instances. The mean approximation ratio is less than $1.05$, and over $90\%$ of instances have an approximation ratio of less than $1.20$ under all network settings. 
(ii) Algorithm \ref{Alg:approximation-2} performs worst on MaxFlow routing, with mean, p90, and maximum approximation ratios larger than those of the other three policies. Notably, under density $0.4$ and $|\mathcal{V}_A|=10$, an instance exists where the approximation ratio reaches $2$, the bound established in Theorem \ref{thm:approximation}. This outcome can be attributed to MaxFlow routing distributing traffic to optimally exploit network capacity, resulting in a considerable portion of traffic being dispersed through paths without adversarial nodes. (iii) The routing attack performs best on density $0.4$ and $|\mathcal{V}_A|=20$ and worst on density $0.8$ and $|\mathcal{V}_A|=10$ in general. This observation aligns with the intuition that higher density combined with fewer adversarial nodes allows network traffic to access more available paths without adversarial nodes, thus reducing the possibility of forming a node cut so that Algorithm \ref{Alg:approximation-2} is able to output the optimal attack as Algorithm \ref{Alg:optimal-metric-1-general}, which echoes Corollary \ref{coro:node_cut}. 

We proceed to compare the performance gaps to the baseline over the 3 algorithms:  Algorithm \ref{Alg:approximation-2}, Algorithm \ref{Alg:heuristic}, and the local heuristic. We illustrate the CDFs of the approximation ratios of the three algorithms under four default routing policies with density $0.8$ and $|\mathcal{V}_A|=20$ in Fig.~\ref{fig:N50_approximation_ratio_different_default_routing}.Algorithm \ref{Alg:approximation-2} outperforms Algorithm \ref{Alg:heuristic} and the local heuristic under all the given default routing policies. The CDF curves reveal that Algorithm \ref{Alg:approximation-2} identifies the optimal solution in over $80\%$ of the instances. 
Moreover, Algorithm \ref{Alg:heuristic} strikes a balance between attack performance and time efficiency, compared with the optimal attack performance of Algorithm \ref{Alg:optimal-metric-1-general} with higher time cost in solving LPs, and the local heuristics with lowest time complexity but without any performance guarantee. We can observe that Algorithm \ref{Alg:heuristic} leads to approximation ratio below $2.0$ in more than $50\%$ of test cases under all these four default routing policies.

\subsection{Loss Maximization}
\label{sec:simulation-loss-maximization}


We evaluate the performance of loss maximization under Algorithm \ref{Alg:max_loss_approx_multiplicative} and \ref{Alg:max_loss_approx_additive} in single-hop networks.  
We simulate various settings in $8\times 8$ and $16\times 16$ single-hop networks, considering all the ingress nodes as adversarial nodes, which does not compromise generality as explained in Section \ref{subsec:apprimaxtion_algorithm}. We compare our proposed algorithms with two heuristics: (i) \textbf{Min}$\mu$: Each adversarial node routes all the traffic to its connected egress node with the minimum service rate. (ii) \textbf{Rand}: Each adversarial node randomly selects a connected egress node and routes all the traffic to it.


We evaluate algorithms across various network settings in three dimensions: (i) Network Density: as in Section \ref{subsec:no_loss_throughput_simulation}. (ii) $(\mu,\lambda)$-ratio, which represents $\sum_{j=1}^{N_D} \mu_j / \sum_{i=1}^{N_S} \lambda_i$, the ratio between the total service rates and traffic arrival rates. A higher ratio implies lighter traffic loads in the network, thus more challenging for applying routing attacks to cause overload. (iii) Uniformity of service rates $\boldsymbol{\mu}$ among all the egress nodes. We consider two scenarios: \emph{heterogeneous} service rates, where the service rates are randomly generated given the $(\mu,\lambda)$-ratio, and \emph{homogeneous} service rates, where the maximum difference of service rates between any pair of egress nodes is within $10\%$, i.e., $\max_{j_1\neq j_2} |\mu_{j_1}-\mu_{j_2}| / \max\{\mu_{j_1}, \mu_{j_2}\} \leq 10\%$, the idea of which is widely adopted in real-world data center networks to avoid speed mismatch \cite{singh2015jupiter,zhang2021gemini}. For each network setting, we evaluate 30 random  topologies between $\mathcal{V}_S$ and $\mathcal{V}_D$, and 50 different values of $\boldsymbol{\lambda}$ and $\boldsymbol{\mu}$ per topology, subject to the given constraints of $(\mu,\lambda)$-ratio and uniformity of $\boldsymbol{\mu}$.


We first present results on $8\times 8$ networks. We validate the proved performance guarantees of Algorithm \ref{Alg:max_loss_approx_multiplicative} and \ref{Alg:max_loss_approx_additive}   , since the brute force mechanism in Proposition \ref{prop:boundary_property_maxloss} is not prohibitive to simulate under $|\mathcal{V}_A|=8$. We present the results in Table \ref{tab:approximation-max-loss}.
For Algorithm \ref{Alg:max_loss_approx_multiplicative}, we assess the approximation ratio $\Delta_{MUL}/\Delta_{OPT}\in[0,1]$, which achieves optimum in more than $90\%$ of the tested instances, and the worst approximation ratio is above $2/3$ under all tested settings, much surpassing the bound $1/\sqrt{|\mathcal{V}_A|}=1/\sqrt{8}=0.35$ in Theorem \ref{thm:multiplicative}. For Algorithm \ref{Alg:max_loss_approx_additive}, we evaluate $(\Delta_{OPT} - \Delta_{ADD}) / \lambda$, which is less than $0.05$ under heterogeneous service rates and close to $0$ under homogeneous service rates in $90\%$ of the instances, with the largest gap being $0.18$, below the bound of $1/4$ in Theorem \ref{thm:additive}.  These results demonstrate the near-optimal performance in most of the tested cases by both Algorithm \ref{Alg:max_loss_approx_multiplicative} and \ref{Alg:max_loss_approx_additive}.


We then compare Algorithm \ref{Alg:max_loss_approx_multiplicative} and \ref{Alg:max_loss_approx_additive} with the heuristic algorithms Min$\mu$ and Rand in Fig.~\ref{fig:Paper_approximation_ratio_density0.3_scale2.0_formal}. Here we only present the CDFs of approximation ratios $\Delta_{ALG}/\Delta_{OPT}$ given density $0.3$ and a $(\mu,\lambda)$-ratio of $2$, where the ranking of the tested algorithms remains the same under the metric of performance gap $(\Delta_{ALG} - \Delta_{OPT})/\lambda$. We have the following observations: (i) Algorithm \ref{Alg:max_loss_approx_multiplicative} and \ref{Alg:max_loss_approx_additive} significantly outperform the other heuristics, achieving approximation ratios close to $1$ under both heterogeneous and homogeneous $\boldsymbol{\mu}$ in most instances. (ii) The advantage of Algorithm \ref{Alg:max_loss_approx_multiplicative} and \ref{Alg:max_loss_approx_additive} over {Min}$\mu$ diminishes under heterogeneous $\boldsymbol{\mu}$ compared to homogeneous $\boldsymbol{\mu}$. This is because under heterogeneous service rates, there are instances where an egress node has very low service rates but is connected to a large number of ingress nodes. {Min}$\mu$ is well-suited to these instances by overloading this egress node, thus closely approximating the optimal attack.

\begin{table}[!htbp]
\begin{adjustbox}{width=\columnwidth}
\begin{tabular}{c@{\hspace{0.5em}}c|cc|cc!{\vrule width 1.5pt}cc|cc}
\toprule
{} & {} & \multicolumn{2}{c|}{Alg. \ref{Alg:max_loss_approx_multiplicative} (Hetero)} & \multicolumn{2}{c!{\vrule width 1.5pt}}{Alg. \ref{Alg:max_loss_approx_multiplicative} (Homo)} & \multicolumn{2}{c|}{Alg. \ref{Alg:max_loss_approx_additive} (Hetero)} & \multicolumn{2}{c}{Alg. \ref{Alg:max_loss_approx_additive} (Homo)} \\
{Den} & {$\mu/\lambda$} &   {$10\%$} &   Min &   $10\%$ &   Min &    $90\%$ &   Max &   $90\%$ &   Max \\
\midrule
0.3 & 2.0 &   1.00 &  0.69 &  1.00 &  0.75 &   0.04 &  0.17 &  0.00 &  0.17 \\
0.5 & 2.0 &   1.00 &  0.68 &  1.00 &  0.74 &   0.04 &  0.14 &  0.00 &  0.14 \\
0.3 & 3.0 &   1.00 &  0.73 &  1.00 &  0.69 &   0.01 &  0.13 &  0.00 &  0.06 \\
0.5 & 3.0 &   1.00 &  0.76 &  1.00 &  0.71 &   0.04 &  0.18 &  0.00 &  0.07 \\
\bottomrule
\end{tabular}
\end{adjustbox}
\caption{Statistics of performance guarantee for \eqref{eqn:max-loss} in $8\times 8$ networks under different network settings ($\mu/\lambda$ means $(\mu,\lambda)$-ratio; $10\%$, $90\%$ mean 10-percentile, 90-percentile)}
\label{tab:approximation-max-loss}
\end{table}

\begin{figure}[!htbp]
\centering
\includegraphics[width=1.0\linewidth]{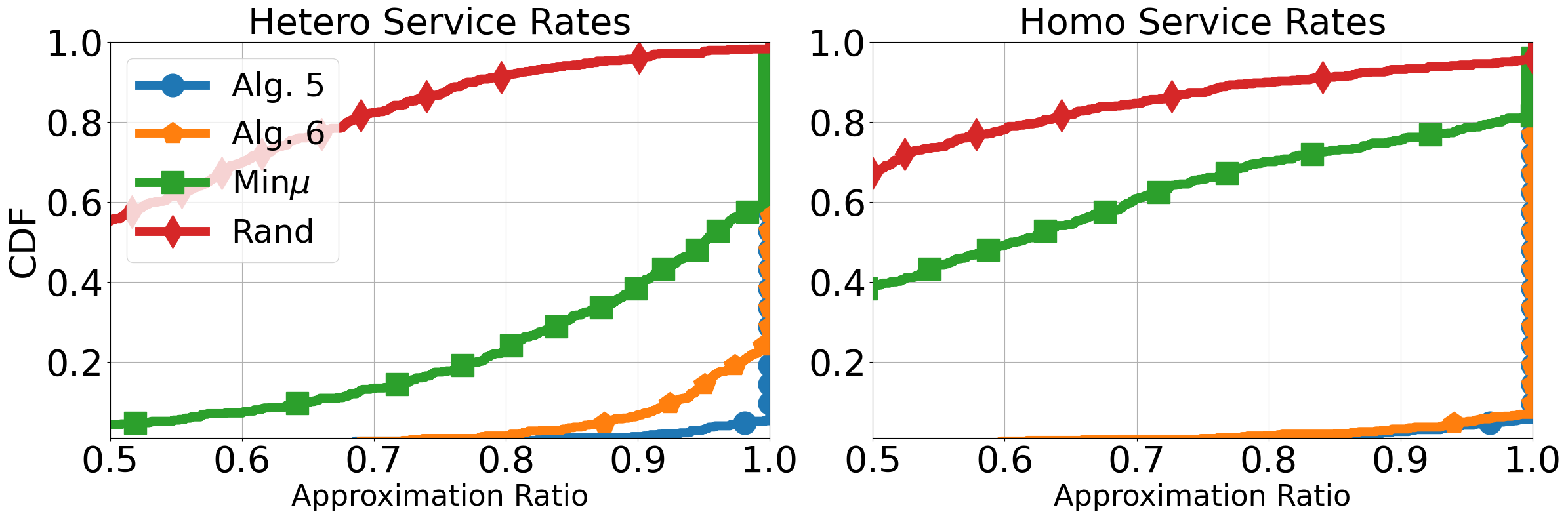}
\caption{CDFs of approximation ratio under density $0.3$ and $\mu/\lambda=2$ in $8\times 8$ networks}
\label{fig:Paper_approximation_ratio_density0.3_scale2.0_formal}
\end{figure}

We further evaluate over larger $16\times 16$ networks. The brute-force approach becomes prohibitive to simulate under $|\mathcal{V}_A|=16$. Therefore we measure the alternative metric \emph{loss ratio}, which is the ratio between the loss and total traffic arrival rate. We present the mean loss ratio results among all the tested instances under different settings in Table \ref{tab:approximation-max-loss-16} with the following observations. (i) {Min}$\mu$ achieves similar performance to our Algorithm \ref{Alg:max_loss_approx_multiplicative} and \ref{Alg:max_loss_approx_additive} under heterogeneous service rates, while far inferior under homogeneous service rates, matching the results in Fig.~\ref{fig:Paper_approximation_ratio_density0.3_scale2.0_formal}. (ii) Under a $(\mu, \lambda)$-ratio of $8$ with homogeneous $\boldsymbol{\mu}$, with high probability there is no way to induce overload by routing attack. However, our proposed Algorithm \ref{Alg:max_loss_approx_multiplicative} and \ref{Alg:max_loss_approx_additive} can still cause overload whenever such possibility exists, while Min$\mu$ and Rand are prone to miss. We further visualize the CDFs of loss ratio under density $0.25$ and a $(\mu,\lambda)$-ratio of $4$ in Fig.~\ref{fig:Paper_loss_ratio_m16_n16_density0.25_scale4.0_formal}, which echoes the results in Table \ref{tab:approximation-max-loss-16}. The evaluation reveals that the proposed Algorithm \ref{Alg:max_loss_approx_multiplicative} and \ref{Alg:max_loss_approx_additive} can effectively overload the networks with near-optimal performance.

\begin{table}[!htbp]
\centering
\begin{adjustbox}{width=0.95\columnwidth}
{\fontsize{10}{12}\selectfont
\begin{tabular}{c@{\hspace{0.5em}}ccccc}
\toprule
     &      & $\mu/\lambda=4$ & $\mu/\lambda=4$ & $\mu/\lambda=8$ & $\mu/\lambda=8$ \\
     &      & Den$=0.5$ & Den$=0.25$ & Den$=0.5$ & Den$=0.25$ \\
\midrule
Hetero & Alg.~\ref{Alg:max_loss_approx_multiplicative} &           0.79 &            0.53 &           0.67 &            0.37 \\
Hetero     & Alg.~\ref{Alg:max_loss_approx_additive} &           0.78 &            0.53 &           0.66 &            0.37 \\
Hetero     & Min$\mu$ &           0.76 &            0.49 &           0.64 &            0.34 \\
Hetero     & Rand &           0.14 &            0.14 &           0.07 &            0.07 \\
    \hline
Homo & Alg.~\ref{Alg:max_loss_approx_multiplicative} &           0.53 &            0.29 &           0.25 &            0.02 \\
Homo     & Alg.~\ref{Alg:max_loss_approx_additive} &           0.53 &            0.28 &           0.25 &            0.02 \\
Homo     & Min$\mu$ &           0.34 &            0.11 &           0.05 &            0.00 \\
Homo     & Rand &           0.01 &            0.01 &           0.00 &            0.00 \\
\bottomrule
\end{tabular}
}
\end{adjustbox}
\caption{Mean loss ratio in $16\times 16$ networks}
\label{tab:approximation-max-loss-16}
\end{table}

\begin{figure}[!htbp]
\centering
\includegraphics[width=1.0\linewidth]{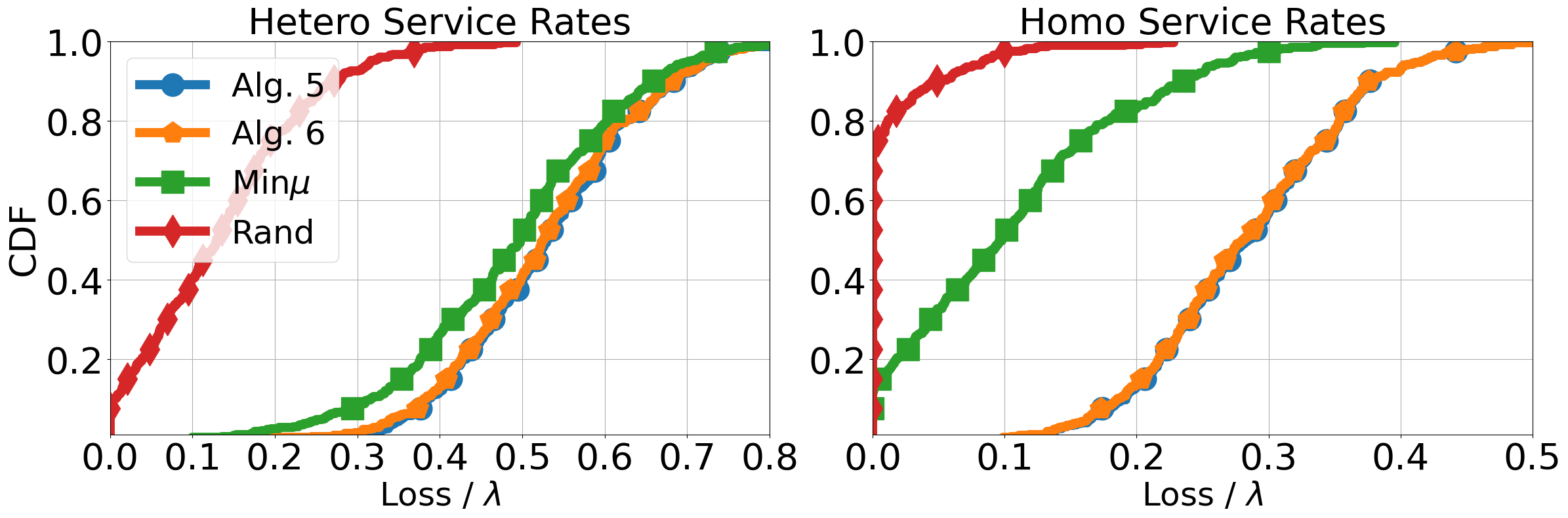}
\caption{CDFs of loss ratio under density $0.25$ and $\mu/\lambda=4$ in $16\times 16$ networks}
\label{fig:Paper_loss_ratio_m16_n16_density0.25_scale4.0_formal}
\end{figure}

In summary, the theoretical optimality of Algorithm \ref{Alg:optimal-metric-1-general} and the empirical near-optimality of Algorithm \ref{Alg:max_loss_approx_multiplicative} and \ref{Alg:max_loss_approx_additive} demonstrate that these proposed routing attack algorithms can accurately approximate the highest overload that a routing attack can induce given an arbitrary set of adversarial nodes.


\subsection{Optimal Node Selection}
\label{sec:simulation_optimal_node_selection}

Algorithm \ref{Alg:optimal_node_selection_parallel} is proved optimal for no-loss throughput minimization. Therefore, here we solely evaluate the performance of the heuristic Algorithm \ref{Alg:optimal_node_selection_max_loss} for loss maximization over different network settings. We consider $400$ randomly generated $12\times 12$ single-hop networks, with the link existence probability being $0.3$ and $(\mu, \lambda)$-ratio being $2$. The candidate node set $\mathcal{V}_{rand}$ includes all $12$ ingress nodes. We set to choose at most $K = 4$ nodes to attack. We consider heterogeneous and homogeneous service rates at the egress nodes as in Section \ref{sec:simulation-loss-maximization}. The total loss depends on the performance of both the node selection, and the routing attack algorithm based on the selected nodes to hijack. We evaluate the following 3 methods: (i) Optimal node selection $+$ Optimal routing attack, where the node selection is to brute-force all $\binom{12}{4}$ choices, and the optimal routing attack is done based on Proposition \ref{prop:boundary_property_maxloss}; (ii) Algorithm \ref{Alg:optimal_node_selection_max_loss} for node selection $+$ Optimal routing attack based on Proposition \ref{prop:boundary_property_maxloss}; (iii) Both node selection and routing attack output by Algorithm \ref{Alg:optimal_node_selection_max_loss}. We calculate the approximation ratio of the total loss under method (ii) and (iii) over the maximum total loss under method (i), and present the CDFs in Fig.~\ref{fig:Paper_m8_n8_k3_density0.3_scale2.0_formal}. Results demonstrate that with Algorithm \ref{Alg:optimal_node_selection_max_loss} for node selection, the adversary can achieve maximum loss in almost $80\%$ of the test cases, and at least $80\%$ of the maximum loss in more than $95\%$ of the test cases, under the optimal routing attack, which demonstrates the high performance of Algorithm \ref{Alg:optimal_node_selection_max_loss} in node selection in most of the network instances. Furthermore, directly applying the output routing attack solution of Algorithm \ref{Alg:optimal_node_selection_max_loss}, which is of polynomial time complexity for general $\mathcal{V}_{rand}$ and $K$, has performance close to applying the optimal routing attack: the adversary can achieve maximum loss in almost $70\%$ of the test cases, and at least $80\%$ of the maximum loss in more than $90\%$ of the test cases. The result demonstrates the good performance of Algorithm \ref{Alg:optimal_node_selection_max_loss} under most of the network settings, showing the high capability of routing attack to maximize network loss, and its usage for network providers to identify the critical nodes to be protected from the routing attack to avoid the risk of high overload. 

\begin{figure}[!htbp]
\centering
\includegraphics[width=1.0\linewidth]{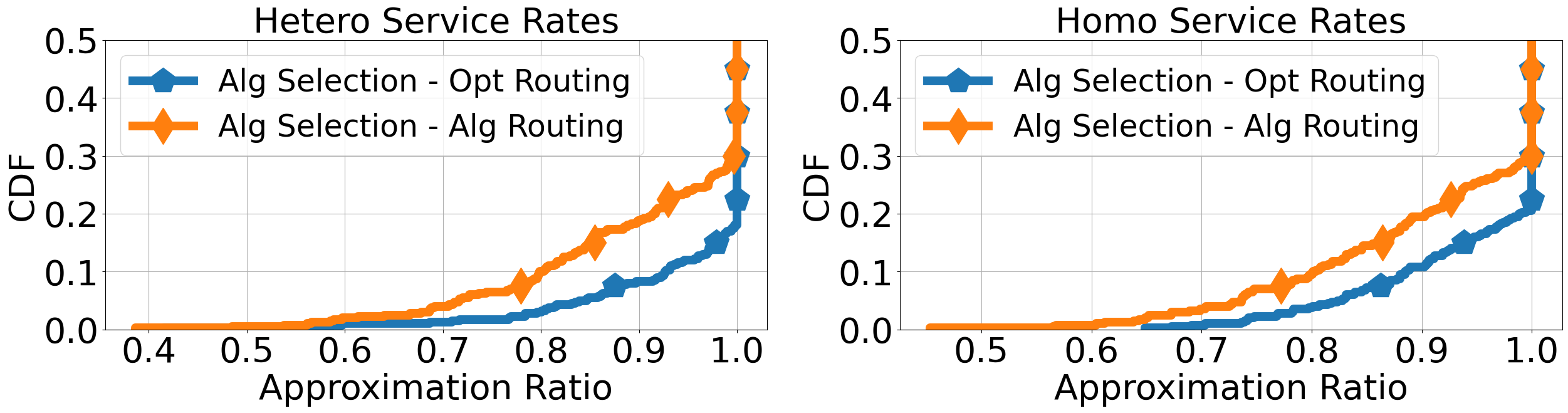}
\caption{CDFs of approximation ratio of network loss under density $0.3$ and a $(\mu, \lambda)$-ratio of $2$ in $16\times 16$ networks}
\label{fig:Paper_m8_n8_k3_density0.3_scale2.0_formal}
\end{figure}

\section{Conclusion}


In this paper, we quantify the threat of routing attacks for causing network overload.
We investigate the optimal routing attacks for {no-loss throughput minimization} and {loss maximization}.
We demonstrate that the no-loss throughput can be minimized in polynomial time in general multi-hop networks. We further develop a 2-approximation algorithm by only leveraging the downstream information of the adversarial nodes. We establish that loss maximization is NP-hard and propose two approximation algorithms with guaranteed performance in single-hop networks. Moreover, we address the adversary's optimal selection of nodes to conduct routing attacks and propose heuristic algorithms for this NP-hard problem. Our performance evaluation showcases the near-optimal performance of the proposed algorithms across a wide range of network settings. Future directions include deriving the performance guarantee of loss maximization in multi-hop networks, investigating the case where normal nodes can adjust their routing in response to routing attacks, 
and designing network control algorithms under routing attacks.


\bibliography{NodeInterdiction}
\bibliographystyle{IEEEtran}

\appendix

\subsection{Proof of Theorem \ref{thm:multiplicative}}

\begin{proof}
Based on Proposition \ref{prop:boundary_property_maxloss}, there exists one optimal routing attack where each adversarial node in $\mathcal{V}_A$ sends all its traffic to an egress node. Consider an optimal routing attack w.l.o.g. where the ingress node $s_i$ sends all traffic to the egress node $d_{j_i}$, i.e., $x_{s_id_{j_i}}=1$. Denote the total overload of the optimal attack as $\Delta_{OPT}$. Multiple ingress nodes may choose the same egress node (i.e., $d_{j_{i_1}}$ and $d_{j_{i_2}}$ may be the same for some $i_1 \neq i_2$), hence we remove the repeated elements in $\{d_{j_i}\}_{i=1}^m$ into $\{d_j\}_{j=1}^{L^{\prime}}$, where $L^{\prime}$ denotes the number of egress nodes that receive traffic from at least one ingress node. Furthermore, among $L^{\prime}$ egress nodes, there may exist nodes with no overload where total ingress is no greater than total egress. We further remove these nodes from $\{d_j\}_{j=1}^{L^{\prime}}$ into $\{d_j\}_{j=1}^{L}$ w.l.o.g., i.e. $L$ is the number of overloaded egress nodes under this optimal routing attack. Note that $L\leq |\mathcal{V}_A|$. In the following, we prove that the approximation ratio under Approach 1 in Algorithm \ref{Alg:max_loss_approx_multiplicative} is at least $1/L$, and that under Approach 2 is at least $L/|\mathcal{V}_A|$. With these proved, taking the solution that causes higher loss between these two methods causes approximation ratio $\max\{1/L, L/|\mathcal{V}_A|\} \geq 1/\sqrt{|\mathcal{V}_A|}$. 

\emph{Approach 1}: Denote the total overload under the output of Approach 1 as $\Delta_{1}$. The first step in Approach 1 is the find $j^*\leftarrow \argmax_{j}\sum_{(s_i,d_j)\in \mathcal{E}} \lambda_i - \mu_j$. If there is no overload at egress node $d_{j^*}$, then $\Delta_{1} =\Delta_{OPT}=0$, i.e., no overload can be caused under any routing attack. When there exists overload, then $\Delta_{1} \geq \sum_{(s_i,d_{j^*})\in \mathcal{E}} \lambda_i - \mu_{j^*} > 0$, while for any overloaded egress node in $\{d_j\}_{j=1}^{L}$, the overload is less than at $d_{j^*}$ under Approach 1. 
Therefore 
$$
\frac{\Delta_1}{\Delta_{OPT}} \geq 
\frac{\sum_{(s_i,d_{j^*})\in \mathcal{E}} \lambda_i - \mu_{j^*}}{L \left(\sum_{(s_i,d_{j^*})\in \mathcal{E}} \lambda_i - \mu_{j^*}\right)} = \frac{1}{L}.
$$

\emph{Approach 2}: We consider when Approach 2 runs one specific iteration to overload a single egress node $d_{j^*}\in \mathcal{V}_D$ according to the definition of $PSO$ in \eqref{eqn:PSO}, and we explain the idea based on the example in Fig.~\ref{fig:Paper_example_approximation_multiplicative_approach2}. Denote the set of connected ingress nodes that send all their traffic to $d_{j^*}$ in this iteration as $\mathcal{S}_{j^*} = \arg\max_{\mathcal{S}} PSO[d_{j^*}]$. In Fig.~\ref{fig:Paper_example_approximation_multiplicative_approach2}, $j^*=3$ and $\mathcal{S}_{3}= \{s_4,s_5,s_6\}$, where the routing is highlighted in blue\footnote{Algorithm will choose the blue routing when $\lambda_3$ is small so that it will lower down the overload-per-ingress.}. We then denote as  $\mathcal{D}_{j^*}$ the set of egress nodes to which the optimal routing attack routes all the traffic in $\mathcal{S}_{j^*}$. In Fig.~\ref{fig:Paper_example_approximation_multiplicative_approach2}, $\mathcal{D}_{j^*}=\mathcal{D}_{3}=\{d_2,d_4\}$, with the optimal routing attack highlighted in red. Then we evaluate the ratio between the loss at $d_{j^*}$ under Approach 2, denoted by $\delta_{d_{j^*}}^{(2)}$, and the total loss over $\mathcal{D}_{j^*}$ under optimal solution $\{\mathbf{x}_i^*\}_{s_i \in \mathcal{V}_A}$, denoted by $\sum_{j\in \mathcal{D}_{j^*}}\delta_{d_j}^{OPT}$, which can be lower bounded as follows
\begin{equation}
\begin{aligned}
&\frac{\delta_{d_{j^*}}^{(2)}}{\sum_{j\in\mathcal{D}_{j^*}}\delta_{d_j}^{OPT}}
= \frac{\left(\sum_{s_i \in \mathcal{S}_{j^*}}\lambda_i\right) - \mu_{j^*}}{\sum_{d_j\in \mathcal{D}_{j^*}}\left(\left(\sum_{s_i: x_{s_id_j}^*=1} \lambda_i \right) - \mu_j\right)}
\\& = \frac{|\mathcal{S}_{j^*}|\frac{\left(\sum_{s_i \in \mathcal{S}_{j^*}}\lambda_i\right) - \mu_{j^*}}{|\mathcal{S}_{j^*}|}}{\sum_{d_j\in \mathcal{D}_{j^*}}|\mathcal{S}_{j}^{OPT}|\frac{\left(\sum_{s_i\in |\mathcal{S}_{j}^{OPT}|} \lambda_i \right) - \mu_j}{|\mathcal{S}_{j}^{OPT}|}}
\\& \overset{(a)}{\geq} \frac{|\mathcal{S}_{j^*}|\frac{\left(\sum_{s_i \in \mathcal{S}_{j^*}}\lambda_i\right) - \mu_{j^*}}{|\mathcal{S}_{j^*}|}}{\sum_{d_j\in \mathcal{D}_{j^*}}|\mathcal{S}_{j}^{OPT}|\frac{\left(\sum_{s_i \in \mathcal{S}_{j^*}}\lambda_i\right) - \mu_{j^*}}{|\mathcal{S}_{j^*}|}} \\&= \frac{|\mathcal{S}_{j^*}|}{\sum_{d_j\in \mathcal{D}_{j^*}} |\mathcal{S}_{j}^{OPT}|} \overset{(b)}{\geq} \frac{|\mathcal{S}_{j^*}|}{|\mathcal{V}_A|} \overset{(c)}{\geq} \frac{L}{|\mathcal{V}_A|}
\end{aligned}
\end{equation}
where inequality $(a)$ holds due the optimality of $d_{j^*}$ among all egress nodes w.r.t. \eqref{eqn:PSO}, while inequality $(b)$ holds due to $\cup_{d_j\in \mathcal{D}_{j^*}} \mathcal{S}_{j}^{OPT} \subseteq \mathcal{V}_A$ (in the example of  Fig.~\ref{fig:Paper_example_approximation_multiplicative_approach2} $\cup_{d_j\in \mathcal{D}_{j^*}} \mathcal{S}_{j}^{OPT} = \mathcal{V}_A$), and $(c)$ holds since each ingress node sends all its traffic to a single egress node, thus the number of overloaded egress nodes $L$ under optimal attack is at most $|\mathcal{S}_{j^*}|$. In Fig.~\ref{fig:Paper_example_approximation_multiplicative_approach2}, $|\mathcal{S}_{j^*}|=|\mathcal{S}_{3}|=3$, $L=2$, and $|\mathcal{V}_A|=7$.

\begin{figure}[!htbp]
\centering
\includegraphics[width=0.85\linewidth]{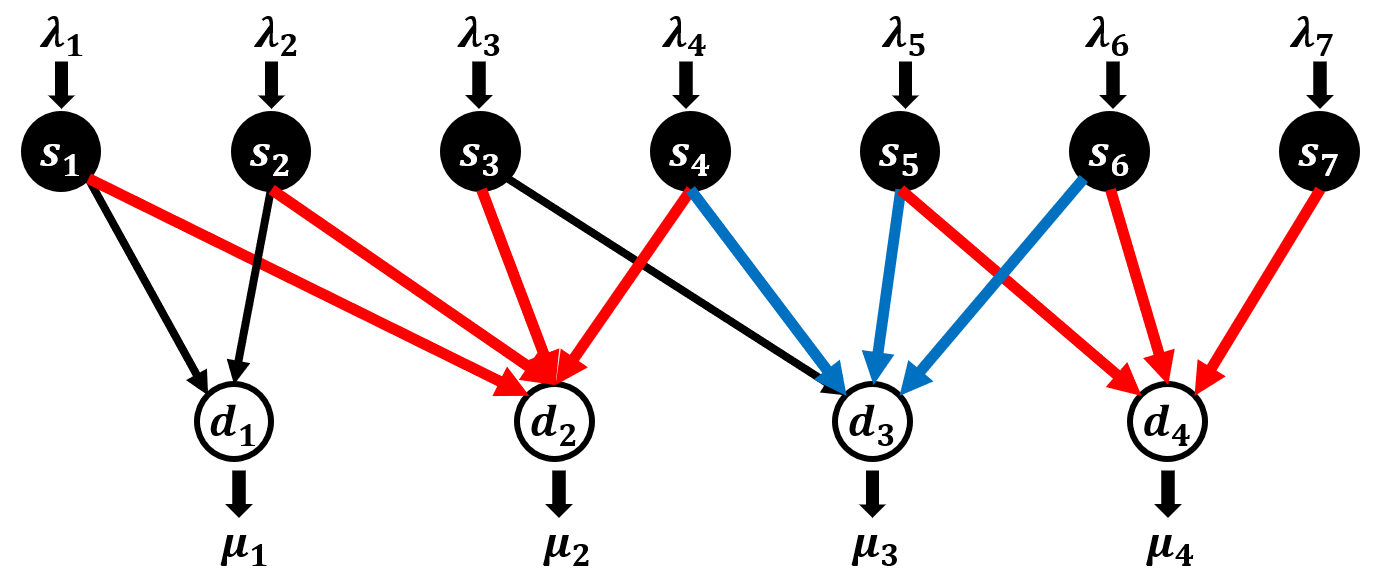}
\caption{Examples for Proof of Approach 2}
\label{fig:Paper_example_approximation_multiplicative_approach2}
\end{figure}

\end{proof}

\subsection{Proof of Theorem \ref{thm:additive}}

\begin{proof}
We conduct the proof through three steps: (i) $2\times 2$ single-hop networks; (ii) general single-hop networks while the optimal routing attack only overloads a single egress node; (iii) extension of results in previous step to general single-hop networks.

\emph{Case 1: $2\times 2$ networks}. Consider the $2\times 2$ example in Fig.~\ref{fig:Paper_example_approximation_additive} with $\mathcal{V}_S = \{s_1,s_2\}$ and $\mathcal{V}_D = \{d_1,d_2\}$, and general $\boldsymbol{\lambda}=(\lambda_1,\lambda_2)$ and $\boldsymbol{\mu}=(\mu_1,\mu_2)$, where we introduce a meta source and destination node equivalently. $s_1$ can only send traffic to $d_1$. The routing of $s_2$ under Algorithm \ref{Alg:max_loss_approx_additive} is determined by the comparison between $\mu_1/(x_1+x_2)=\mu_1$ and $\mu_2/x_2$: Sending all traffic to $d_1$ if $\mu_1 < \mu_2/x_2$ else $d_2$. If $\mu_1\leq\mu_2/x_2$, then $\Delta_{ALG}$, the loss by Algorithm \ref{Alg:max_loss_approx_additive} is $[\lambda - \mu_1]^+$, and in this case the optimal solution has the same loss $\Delta_{OPT}=\Delta_{ALG}$. If $\mu_1> \mu_2/x_2$, then the $\Delta_{ALG}= [\lambda x_1 - \mu_1]^+ + (\lambda x_2 - \mu_2) $ where $\lambda x_2 > \mu_2$ otherwise node $s_2$ will not route to $d_2$. The only possible case that $\Delta_{OPT} > \Delta_{ALG}$ is that the optimal solution is $x_{s_2d_1}=1$. In this case, we have the following gap
$$
\begin{aligned}
&\frac{\Delta_{OPT}-\Delta_{ALG}}{\lambda}
= \frac{(\lambda-\mu_1) - [\lambda x_1 - \mu_1]^+ - (\lambda x_2 - \mu_2)}{\lambda}
\\&= \frac{\mu_2 + \lambda x_1 - \mu_1 - [\lambda x_1 - \mu_1]^+}{\lambda}
\overset{(i)}{\leq} \frac{\mu_2}{\mu_1}x_1 \overset{(ii)}{\leq} x_1x_2 \overset{(iii)}{\leq} \frac{1}{4} 
\end{aligned}
$$
where $(i)$ holds as $\lambda = \frac{\mu_1}{x_1}$ maximizes the gap, which is
$x_1 - \frac{\mu_1-\mu_2}{\lambda}$ when $\lambda < \frac{\mu}{x_1}$ and $ 
\frac{\mu_2}{\lambda}$ and $\lambda \geq \frac{\mu}{x_1}$, 
and note that $\mu_1 > \mu_2/x_2 \geq \mu_2$.
$(ii)$ holds due to $\mu_1> \mu_2/x_2$, and $(iii)$ holds due to $x_1 + x_2 = 1$, thus $x_1x_2\leq 1/4$.

\emph{Case 2: $M\times N$ networks with the optimal routing overloading one egress node}. We extend the example in Fig.~\ref{fig:Paper_example_approximation_additive} to general $M=K$ ingress nodes and $N=K$ egress nodes w.l.o.g., where all the ingress nodes are connected to $d_1$, while $s_i$ is at least connected to $d_i$. We consider the extreme case that will cause biggest gap between the optimal solution and the output from Algorithm \ref{Alg:max_loss_approx_additive}, where the optimal routing is $x_{s_id_1}^*=1$, while the algorithm outputs the solution $x_{s_id_i}=1$, $\forall i=1,\cdots,K$, under the condition $\frac{\mu_{K}}{x_K} \leq \frac{\mu_{K-1}}{x_{K-1}} \cdots \leq \frac{\mu_2}{x_2} \leq \frac{\mu_1}{\sum_{i=1}^K x_i}=\mu_1$. This can be derived by induction based on Case 1. Then we have the bound
$$
\begin{aligned}
&\frac{\Delta_{OPT}-\Delta_{ALG}}{\lambda}
= \frac{\left(\lambda \sum_{i=1}^K x_i - \mu_1\right) - \sum_{j=1}^K (\lambda x_j - \mu_j)}{\lambda}
\\&= \frac{\lambda x_1 - \mu_1 - [\lambda x_1 - \mu_1]^+ + \sum_{j=2}^K \mu_j}{\lambda}
\\&\leq \frac{\sum_{j=2}^K \mu_j}{\mu_i}x_1
\leq x_1(\sum_{i=2}^Kx_i) = x_1(1-x_1) \leq \frac{1}{4}
\end{aligned}
$$

\emph{Case 3: Special case of $M\times N$ networks with the optimal routing attack overloading multiple egress nodes}. We now extend to multiple overloaded egress nodes under optimal routing attack. Here we first consider a special case in Fig.~\ref{fig:Paper_Additive_Approximation_Proof_2}, where the optimal routing attack overloads two egress nodes $d_1$ and $d_4$ highlighted in red, while Algorithm \ref{Alg:max_loss_approx_additive} outputs the routing attack as highlighted in blue, under the conditions that $\mu_2 / x_2, \mu_3 / x_3 \leq \mu_1 / (x_1 + x_2 + x_3)$ and $\mu_5 / x_5 \leq \mu_4 / (x_4 + x_5)$. Based on the results in Case 2, we can derive
$$
\begin{aligned}
&\frac{\Delta_{OPT}-\Delta_{ALG}}{\lambda}
= \frac{1}{\lambda}\left( \lambda x_1 - \mu_1 - [\lambda x_1  - \mu_1]^+  + \mu_2 + \mu_3\right)
\\& \qquad \qquad + \frac{1}{\lambda}\left( \lambda x_4 - \mu_4 - [\lambda x_4  - \mu_4]^+  + \mu_5 \right)
\\& \leq \frac{x_1+x_2+x_3}{\lambda (x_1+x_2+x_3)}\left( \lambda x_1 - \mu_1 - [\lambda x_1  - \mu_1]^+  + \mu_2 + \mu_3\right)
\\& \quad + \frac{x_4+x_5}{\lambda(x_4+x_5)}\left( \lambda x_4 - \mu_4 - [\lambda x_4  - \mu_4]^+  + \mu_5 \right)
\\& 
\overset{(a)}{\leq} \frac{1}{4}(x_1+x_2+x_3) + \frac{1}{4}(x_4+x_5)
= \frac{1}{4}
\end{aligned}
$$
where $(a)$ is based on the result from Case 2. We will use this result in the following discussion over general single-hop networks in Case 4.

\begin{figure}[!htbp]
\centering
\includegraphics[width=0.8\linewidth]{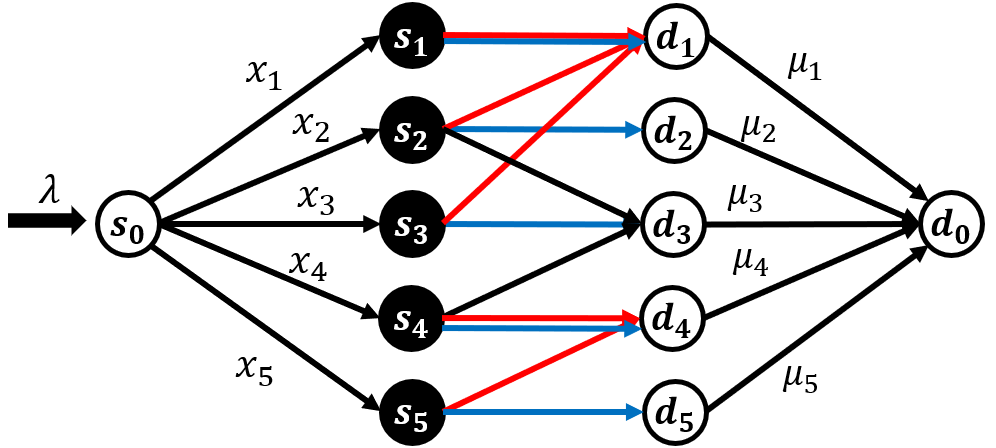}
\caption{Proof for Case 3 of Theorem \ref{thm:additive}}
\label{fig:Paper_Additive_Approximation_Proof_2}
\end{figure}

\emph{Case 4: General single-hop networks}. The optimal attack may overload multiple egress nodes. In this case, we can transform the original graph so that results of Case 2 above can be applied. We give an example in Fig.~\ref{fig:Paper_Additive_Approximation_Proof_3}. Suppose that the optimal routing attack as highlighted in red, while the routing attack from Algorithm \ref{Alg:max_loss_approx_additive} that minimizes no loss throughput $\lambda^*$ is to route traffic from $s_2, s_3$ and $s_4$ to $d_3$, highlighted in blue. Given the routing attack output by Algorithm \ref{Alg:max_loss_approx_additive}, denoted by $\mathbf{x}_{\mathcal{A}}^{ALG}$ we show the following transformation of network topology will keep the same overload as $\mathbf{x}_{\mathcal{A}}^{ALG}$: For each overloaded egress node $d_j$ under $\mathbf{x}_{\mathcal{A}}^{ALG}$, suppose that it receives traffic from $\mathcal{S}_j = \{s_i\}_{i:x_{ij}^{ALG}=1}$, then we decompose node $d_j$ into $|\mathcal{S}_j|$ nodes denoted by $\{d_j^{(k)}\}_{k=1}^{|\mathcal{S}_j|}$, each connect to the $k$-th ingress node in $\mathcal{S}_j$ denoted by $s^{(k)}$. The new service rate for node $d_j^{(k)}$ is $\mu_{j}^{(k)} = \frac{x_{s^{(k)}}}{\sum_{k=1}^{|\mathcal{S}_j|} x_{s^{(k)}}}\mu_j$. It is easy to verify that applying Algorithm \ref{Alg:max_loss_approx_additive} to the transformed network outputs a routing attack solution leads to same overload, as the minimum $\lambda^*$ that saturates $d_j^{(k)}$, $\forall k$ is the same, equal to the $\lambda^*$ that saturates $d_j$ in the original graph. For example in Fig.~\ref{fig:Paper_Additive_Approximation_Proof_3}, node $d_3$ is decomposed into 3 new nodes with new service rates denoted by $\mu_3^{(2)}, \mu_3^{(3)}$ and $\mu_3^{(4)}$.
Note that this transformation is only for proof, where in Algorithm \ref{Alg:max_loss_approx_additive} we use the original graph.

\begin{figure}[!htbp]
\centering
\includegraphics[width=1.0\linewidth]{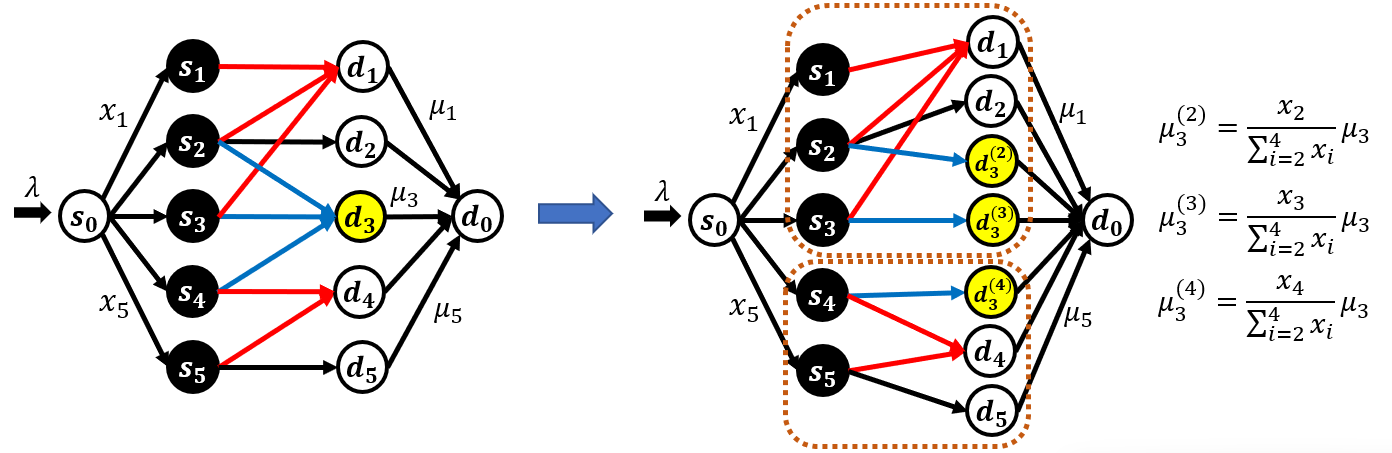}
\caption{Proof for Case 4 of Theorem \ref{thm:additive}}
\label{fig:Paper_Additive_Approximation_Proof_3}
\end{figure}

With the above transformation into multiple basic units in Case 1 (the dashed boxes in Fig.~\ref{fig:Paper_Additive_Approximation_Proof_3}), we can directly apply the analysis for Case 2, and obtain the upper bound of $1/4$.
\end{proof}

\subsection{Proof of Theorem \ref{thm:np_hard_node_selection}}

\begin{proof}
Under $K=O(1)$, the following brute-force algorithm outputs the optimal solution: Enumerate all $\binom{|\mathcal{V}_{cand}|}{K}$ combinations of nodes to attack, and under each combination apply brute force method for $\lambda^*$ minimization given by Theorem \ref{thm:boundary_property}, and apply Proposition \ref{prop:boundary_property_maxloss} for loss maximization. The time complexity are both $O\left(\binom{|\mathcal{V}_{cand}|}{K}\times NK^{d_{\max}} \right) = O(|\mathcal{V}_{cand}|^K K^{d_{\max}} N)$ where $d_{\max}$ is the maximum degree among the $K$ selected nodes to attack.

Under $K=O(N)$, we reduce Set Cover to the problem of node selection for $\lambda^*$-minimization, and the extension to loss maximization is trivial. Given an instance of Set Cover as done in the proof of Theorem \ref{thm:np_hard}, we construct the graph as in Fig.~\ref{fig:Paper_NP_hard_nodeselection}, which adds a node $T$ compared with Fig.~\ref{fig:Paper_NP_hard_maxloss} and the added links are highlighted. The candidate node set is $\{s_i\}_{i=1}^m \cup \{d_j\}_{j=1}^n$, and their default routing policies are all dispatching all the traffic to the intermediate node $T$. All links have unbounded capacity except link $(d_0,T)$ with $c_{d_0T}=1$. Without routing attack, the no-loss throughput $\lambda^*=\infty$. We show that if we can solve the decision problem of node selection: \emph{Can we reduce $\lambda^*$ from $\infty$ to $1$ by attacking $m+k$ nodes?}, then there exists a polynomial-time algorithm to decide if there exists $k$ sets in $\{d_i\}_{i=1}^n$ that can cover all elements corresponding to $\{s_i\}_{i=1}^m$. To reduce $\lambda^*$ to $1$, all traffic flows need to go through link $(d_0,T)$ to node $T$, thus all the $m$ ingress nodes must be attacked to adjust their routing to one of egress nodes instead of directly to $T$. Then if attacking $k$ more nodes in $\{d_j\}_{j=1}^n$ can reduce $\lambda^*$ to $1$, then it guarantees that all the ingress nodes route traffic to one of the $k$ attacked egress nodes, and these $k$ nodes routing all the traffic to $d_0$, which is equivalent to being able to using $K$ sets to cover all elements corresponding to $\{s_i\}_{i=1}^m$.

\begin{figure}[!htbp]
\centering
\includegraphics[width=0.85\linewidth]{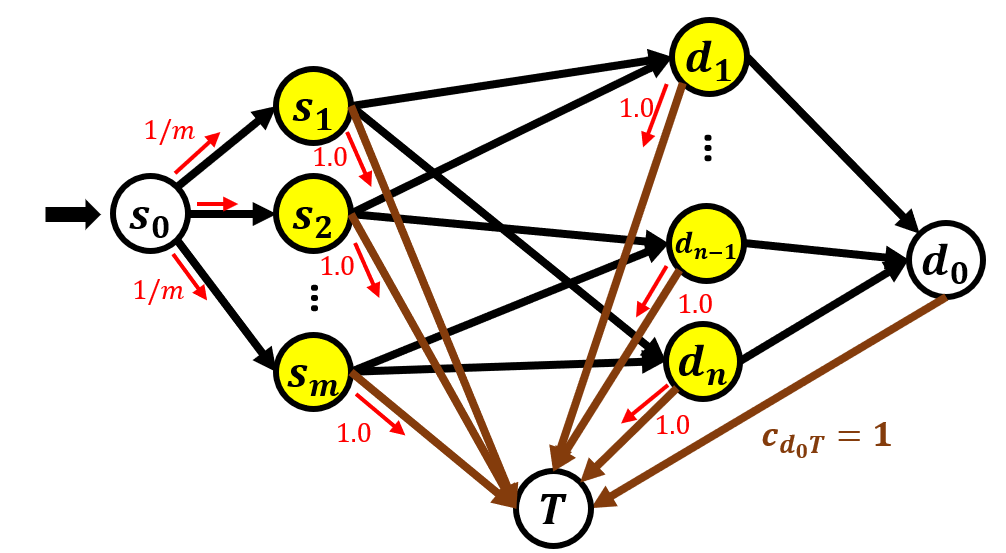}
\caption{NP-hardness of Node Selection Problem}
\label{fig:Paper_NP_hard_nodeselection}
\end{figure}

\end{proof}

\subsection{Proof of Theorem \ref{thm:node_selection_parallel}}

\begin{proof}
Algorithm \ref{Alg:optimal_node_selection_parallel} goes through each downstream link $(i,j)$ to $\mathcal{V}_{cand}$ (including $\mathcal{V}_{cand}$ themselves), and find the optimal choice of $K$ nodes to attack to minimize the arrival rate that saturates one of them $(i,j)$. Since all candidate nodes are parallel, we can separately evaluate the contribution of attacking a node $v\in \mathcal{V}_{cand}$ over $(i,j)$, denoted by $\Delta_v^{(i,j)}$, which is the difference between the arrival rate that saturates link $(i,j)$ with and without attack. We calculate $\Delta_v^{(i,j)}, \forall v \in \mathcal{V}_A^{up,i}$, and then sort the value over all $\mathcal{V}_A^{up,i}$ non-decreasingly. Denote $\Delta_v^i = \max_{(i,j)\in \mathcal{E}} \Delta_v^{(i,j)}$.
Then we can get the optimal choice to each node $i$ as follows, where the choices to each link starting from node $i$ are identical: If $|\mathcal{V}_A^{up,i}| > K$, we choose the top-$K$ nodes in $\mathcal{V}_A^{up,i}$ with highest $\Delta_v^{i}$ values; If $|\mathcal{V}_A^{up,i}| \leq K$, we choose all $|\mathcal{V}_A^{up,i}|$ candidate nodes to attack, where choosing $K-|\mathcal{V}_A^{up,i}|$ candidate nodes does not affect the saturation level of link $(i,j)$ since they are not upstream to node $i$. Denote the no-throughput loss under the choice $\mathcal{V}_A^{(i)}$ as $\lambda_{(i)}^*$, then the optimal node selection to attack is $\mathcal{V}_A^{(i^*)} = \arg\min_{i^*\in \mathcal{V}_A^{down}} \lambda_{(i)}^*$.
\end{proof}

\end{document}